\documentclass{ieeeaccess}
\usepackage{cite}
\usepackage{amsmath,amssymb,amsfonts}
\usepackage{graphicx}
\usepackage{textcomp}
\usepackage{verbatim}
\usepackage[font=small,labelfont=bf,tableposition=top]{caption}
\DeclareCaptionType[fileext=los,placement={th!}]{procedure}

\usepackage{array} %centering tabular with p{}
\usepackage{cite}
\usepackage{multirow} % multirow
\usepackage{algorithm, algorithmic}
\usepackage{graphicx}
\usepackage{subcaption}
\usepackage{verbatim}
\usepackage[table]{xcolor}

\usepackage{multicol}
\usepackage{amsthm} % proof environment
\newtheorem{lemma}{Lemma}
\newtheorem{definition}{Definition}
\newtheorem{corollary}{Corollary}
\newtheorem{theorem}{Theorem}

\def\BibTeX{{\rm B\kern-.05em{\sc i\kern-.025em b}\kern-.08em
    T\kern-.1667em\lower.7ex\hbox{E}\kern-.125emX}}
\begin{document}
%\history{Date of publication xxxx 00, 0000, date of current version xxxx 00, 0000.}
\doi{10.1109/ACCESS.2020.2986131}

\title{Optimal Mapper for OFDM with Index Modulation: A Spectro-Computational Analysis}

\author{\uppercase{Saulo Queiroz}\authorrefmark{1,3}, %\IEEEmembership{Fellow, IEEE},
\uppercase{Joao P. Vilela\authorrefmark{2}, and Edmundo Monteiro}\authorrefmark{3}
\IEEEmembership{Senior Member, IEEE}
}
\address[1]{Academic Department of Informatics,
          Federal University of Techonology (UTFPR), Ponta Grossa, PR, Brazil.
         (e-mail: sauloqueiroz@utfpr.edu.br)}
\address[2]{Department of Computer Science, University of Porto, Portugal.
         (e-mail: jvilela@fc.up.pt)}
\address[3]{CISUC and Department of Informatics Engineering,
University of Coimbra, Portugal 
(e-mail: \{saulo, jpvilela, edmundo\}@dei.uc.pt)}
\tfootnote{This work is partially supported by the European Regional Development Fund (FEDER), through the Regional Operational Programme of Lisbon (POR LISBOA 2020) and the Competitiveness and Internationalization Operational Programme (COMPETE 2020) of the Portugal 2020 framework [Project 5G with Nr. 024539 (POCI-01-0247-FEDER-024539)], and is also partially supported by the CONQUEST project - CMU/ECE/0030/2017 Carrier AggregatiON between Licensed Exclusive and Licensed Shared Access FreQUEncy BandS in HeTerogeneous Networks with Small Cells.}

\markboth
{Saulo Queiroz \headeretal: Optimal Mapper for OFDM with Index Modulation: A Spectro-Computational Analysis}
{Saulo Queiroz \headeretal: Optimal Mapper for OFDM with Index Modulation: A Spectro-Computational Analysis}

\corresp{Corresponding author: Saulo Queiroz (e-mail: sauloqueiroz@utfpr.edu.br).}

\begin{abstract}
In this work, we present an optimal mapper for OFDM with index modulation (OFDM-IM). 
By optimal we mean the mapper achieves the lowest possible asymptotic computational 
complexity (CC) when the spectral efficiency (SE) gain over OFDM maximizes. 
We propose the spectro-computational (SC) analysis to capture the trade-off between
CC and SE and to demonstrate that an $N$-subcarrier OFDM-IM mapper must run in exact 
$\Theta(N)$ time complexity.
We show that an OFDM-IM mapper running faster than such complexity cannot reach the maximal 
SE whereas one running slower nullifies the mapping throughput for arbitrarily large $N$.
We demonstrate our theoretical findings by implementing an open-source library that supports
all DSP steps to map/demap an $N$-subcarrier complex frequency-domain OFDM-IM symbol.
Our implementation supports different index selector algorithms and is the first to
enable the SE maximization while preserving the same time and space asymptotic
complexities of the classic OFDM mapper.
\end{abstract}


\titlepgskip=-15pt

\maketitle
\begin{comment}
\begin{table}[]
\centering
\begin{tabular}{rl}
\multicolumn{2}{c}{\color{accessblue} GLOSSARY} \\
ASIC: & Application-Specific Integrated Circuit \\
BER: & Bit Error Rate \\
CC: & Computational Complexity \\
DFT: & Discrete Fourier Transform \\
DSP: & Digital Signal Processing \\
FPGA: & Field Programmable Gate Array\\
IDFT: & Inverse Discrete Fourier Transform \\
IM: & Index Modulation \\
IxS: & Index Selector \\
LUT: & Look-Up Table\\
OFDM: & Orthogonal Frequency Division Multiplexing \\
OFDM-IM: & OFDM with IM \\
Prop: & Proposed OFDM-IM (De)Mapper \\
Orig: & Original OFDM-IM (De)Mapper \\
RAM: & Random-Access Machine \\
SC: & Spectro-Computational \\
SCE: & Spectro-Computational Efficiency \\
SE: & Spectral Efficiency\\
SP: & Subblock Partitioning
\end{tabular}
\end{table}
\end{comment}
\begin{table}[]
\centering
\begin{tabular}{rl}
\multicolumn{2}{c}{\color{accessblue} NOTATION} \\
$c_i$: & Index of the $i$-th active subcarrier in the symbol \\
$g$: & Number of subblocks per symbol \\
$k$: & Number of active subcarriers \\
$m$: & Total number of bits per symbol \\
$m(N)$: & Asymptotic number of bits per symbol as function of $N$ \\
$n$: & Number of subcarriers per subblock \\
$p$: & Total number of bits per subblock \\
$p_1$: & Number of index modulation bits per subblock\\
$p_2$: & Number of bits per active subcarriers in a subblock\\
$\delta$: & Half-width of the confidence interval\\
$x$: & Number of samples of the steady-state mean\\
\textbf{s}: & List of baseband samples per symbol\\
\textbf{s}$_\beta$: & List of baseband samples in the $\beta$-th subblock\\
$A_{N,k}$: & $N\times k$ Johnson association scheme\\
$I$: & List of active subcarrier indexes per symbol\\
$I_\beta$: & List of active subcarrier indexes in the $\beta$-th subblock\\
$N$: & Number of subcarriers per symbol  \\
$M$: & Constellation size of the modulation diagram\\
$P_1$: & Number of index modulation bits per symbol\\
$P_2$: & Number of bits per active subcarriers in a symbol\\
$X$: & Decimal representation of the $P_1$-bit mapper input\\
$T(N)$: & (De)Mapper computational complexity as function of $N$ \\
$m(N)/T(N)$: & (De)Mapper spectro-computational throughput \\
${N\choose k}$: & $N!/(k!(N-k)!)$ \\
$\kappa$: & Wall-clock runtime of a computational instruction\\
$o(f)$: & Order of growth asymptotically smaller than $f$ \\
$\omega(f)$: & Order of growth asymptotically larger than $f$ \\
$O(f)$: & Order of growth asymptotically equal or smaller than $f$ \\
$\Omega(f)$: & Order of growth asymptotically equal or larger than $f$ \\
$\Theta(f)$: & Order of growth asymptotically equal to $f$ \\
$Z^T$: & Transpose of the matrix $Z$\\
\end{tabular}
\end{table}
\section{Introduction}\label{sec:introduction}
\PARstart{I}ndex Modulation (IM) is a physical layer technique that can improve
the spectral efficiency (SE) of OFDM. IM's basic idea for OFDM~\cite{pc-ofdm-1999,basar-ofdmim-globecom-2012} 
consists in activating $k\in[1,N]$ out of  $N$ subcarriers of the symbol to enable 
extra ${N \choose k}=N!/(k!(N-k)!)$ waveforms.
Of these, OFDM-IM employs $2^{\lfloor\log_2C(N,k)\rfloor}$ to map 
$P_1=\lfloor\log_2{{N} \choose{k}}\rfloor$ bits.  
Besides, modulating the $k$ active subcarriers with an $M$-ary constellation, 
the OFDM-IM symbol can transmit more $P_2=\log_2M$ bits along with $P_1$.
Thus, the OFDM-IM mapper takes a total of $m=P_1+P_2$ bits as input and gives
$k$ complex baseband samples as output for the modulation of the $k$ subcarriers.
In this process, the index selector (IxS) determines the $k$-size list of indexes 
-- out of $2^{P_1}$ possibles -- from the $P_1$-bit input. The remainder $N - k$ 
subcarriers are nullified. The other DSP steps follow as usual in OFDM, except for 
the signal detector at the receiver. In this sense, several research efforts have 
been done to improve the receiver's bit error rate at low computational
complexity~\cite{hu-lowcomp-ofdmim-wcl-2018, sandell-Mlineardetect-commlet-2016, 
siddiq-all_patterns-lettes-2016, zheng-lc_ofdmgim2-2015, basar-ofdmim-transac-2013}. 
Since our focus is on the OFDM-IM mapper, 
we refer the reader to the survey works~\cite{mao-survey-ofdmim-2018, 
basar-thesurvey-2017, sugiura-ofdm_im_spacetimefreq_survey-2017, 
ishikawa-will_OFDMIM_work-2016} for other aspects of the index modulation technique.

\subsection{Problem}
In this work, we concern about whether the OFDM-IM mapper can reach the maximal SE gain 
over its OFDM counterpart keeping the same computational complexity (CC) asymptotic constraints. 
The SE maximization of OFDM-IM over OFDM happens when the IM technique is 
applied on all $N$ subcarriers of the symbol with $k=N/2$ and the active subcarriers 
are BPSK-modulated, i.e., $M=2$~\cite{fan-ofdm_gim1-globecom-2013,fan-ofdmgim3-2015}. 
We refer to this setup as the optimal OFDM-IM configuration. 

The computational complexity of the OFDM-IM mapper under the optimal SE configuration
has been conjectured as an ``impossible task''~\cite{lu-compressiveim-2018,basar-thesurvey-2017}.
This belief comes from the fact that the number of mappable OFDM-IM waveforms grows as fast as
$O({N\choose k})$, which becomes exponential if the optimal SE configuration 
is allowed. Indeed, according to the theory of computation,
a problem of size $N$ is computationally intractable if its time complexity lower bound
is $\Omega(2^N)$. Despite that, as far as we know, \emph{the CC lower
bound required to sustain the maximal SE gain of OFDM-IM remains an open question across
the literature}. Consequently, no prior work can answer whether the OFDM-IM mapper indeed 
needs more asymptotic computational resources than its OFDM counterpart to sustain 
the maximal SE gain.

\subsection{Related Work}
In this subsection, we review the literature related to the design and computational
complexity of the OFDM-IM mapper. 

\subsubsection{Early Attempt}
The earliest mapper for OFDM-IM we find is  due to~\cite{pc-ofdm-1999}.
The authors suggest a Look-Up Table (LUT) to map $P_1$ bits into one
out of  $2^{P_1}$ unique waveforms for relatively small $P_1$.
To avoid the exponential increase in storage implied by the optimal SE 
configuration, the  authors employ a Johnson association scheme~\cite{jas-1978} 
to map $P_1$ based on the recursive matrix 
$A_{N,k}=[ [1$ $0]^T [A_{N-1,k-1}$~$A_{N-1,k}]^T]$,
in which $Z^T$ is the transpose of a given matrix $Z$. Those authors remark that the 
matrix indexes decrease linearly with $N$ towards the base case of recursion.
However, we remark that the overall CC to write
all rows of $A_{N,k}$ is exponential under the optimal SE configuration.
To verify that, consider firstly that $A_{N,k}$ can be lower-bounded 
by $A_{k,k}$, since $k\leq N$. To build $A_{k,k}$, one needs at least two computational
instructions to write the numbers $1$ and $0$ and two other independent and 
distinct recursive calls $A_{k-1,k-1}$ and $A_{k-1,k}$. In the worst-case analysis,
the number of computational steps $T$ to write all entries of $A_{k,k}$
can be captured by the recurrence $T(k)=2+2T(k-1)$, which is trivially verified as $\Omega(2^k)$.
Under the optimal SE setup, the proposed recursive scheme is $\Omega(2^N)$.

\subsubsection{Sub-block Partitioning}
To handle the OFDM-IM mapping overhead, Basar et al.~\cite{basar-ofdmim-transac-2013, 
basar-ofdmim-globecom-2012} propose the subblock partitioning (SP) approach. According
to the survey work of~\cite{basar-thesurvey-2017},  SP and the IxS algorithm presented 
by~\cite{basar-ofdmim-transac-2013, basar-ofdmim-globecom-2012} were (along with a low 
complexity detector) the distinctive methods responsible to release the true potential of 
the IM scheme, thereby shaping the family of index modulation waveforms as we know today.
The key idea of SP is to attenuate the mapper CC by restricting the application of 
the IM technique to smaller portions of the symbol called ``subblocks''. 
The length $n=\lfloor N/g\rfloor$ of each 
subblock depends on the number $g$ of subblocks, which is a configuration parameter of 
OFDM-IM. Increasing $g$, decreases $n$, which causes the complexity of the IxS algorithm 
to decrease too. This way, SP introduces a trade-off between SE and CC, since
the number of OFDM-IM waveforms increases for lower $g$~\cite{basar-ofdmim-transac-2013, 
basar-ofdmim-globecom-2012}. Thus, setting $g=1$  (i.e., deactivating SP) means maximizing 
the SE efficiency. SP has represented the state of the art approach to balance SE and CC 
across the family of IM-based multi-carrier waveforms~\cite{yoon-hammingmapper-2019,
li-opportunisticofdmim-2019, li-layered-2019, newdesign-dm-kim-2019, shi-ofdm-aim-2019, jaradat-number-modulation-2018, 
mao-survey-ofdmim-2018, lu-compressiveim-2018,adaptive-dm-2018, wen-mmode-2018, wen-g-mm-ofdm-im-2018, mao-dm_ofdm_im-ieeeaccess-2017,wen-mmode_ofdm_im-tran_comm-17,ozturk-gfdm_im-ieee_access-17,dm-ofdm-im-cpa-2017,
mao-gdm_ofdm_im-commletters-2017, gokceli-realofdmimieeeaccess-2017, 
basar-thesurvey-2017,fan-ofdmgim2-iet-2016, fan-ofdm_gim1-globecom-2013, fan-ofdmgim3-2015}.

\subsubsection{(Un)Ranking Algorithms}
The IxS algorithm is a mandatory part for the asymptotic analysis of the OFDM-IM mapper.
As observed by authors in~\cite{basar-ofdmim-transac-2013, basar-ofdmim-globecom-2012},
the IxS task at the OFDM-IM transmitter (receiver) can be implemented as an unranking 
(ranking) algorithm. 
By reviewing the literature in combinatorics, one can find out several different 
(un)ranking algorithms, running at different time complexities~\cite{parque-2018, 
shimizu-unrakingsmalk-2014, combinadic-2014, molinero-2001, kreher-1999, kokosinski-1995, 
chen-parallel-1986, er-linearunranking-1985, buckles-comb-1977}. At a first glance,
building the optimal OFDM-IM mapper may just be a matter of adopting the IxS algorithm that
establishes the complexity upper-bound for the (un)ranking problem, i.e., the 
fastest currently known algorithm. However, in the particular domain of OFDM-IM, 
\emph{$k$ represents a trade-off between SE and CC}. Thus, because the literature in
pure combinatorics does not concern about SE as a performance indicator, 
it does not suffice to guide the design of an optimal OFDM-IM mapper.
%This is the case, for instance, of \cite{kokosinski-1995}, that recommends $O(N)$ 
%unranking algorithms only for domains where unranking is occasional. By contrast,
%in the OFDM-IM domain, the IxS algorithm runs in a per-symbol basis, meaning 
%it is performed several times in $\mu s$ temporal scale. As we show throughout this 
%work, the optimal SE cannot be reached unless the mapper complexity runs in
%exact $\Theta(N)$ order of growth.
Therefore, to the best of our knowledge, \emph{no prior analysis concerns about
the OFDM-IM mapper complexity minimization under the constraint of the SE maximization}. 

\subsubsection{Novel SP-Free OFDM-IM Mappers}
In~\cite{salah-2019}, the authors propose the concept of sparsely indexing modulation
to improve the trade-off between SE and energy efficiency of OFDM-IM. Because this concept 
imposes $k$ to be much less than $N$, the authors rely on~\cite{kokosinski-1995} 
to perform IxS in $O(k\log N)$ time. With the achieved time complexity reduction, 
the authors present the first SP-free OFDM-IM mapper. However, the constraint on the 
value of $k$ prevents the $SE$ maximization. 
To identify the largest tolerable computational complexity to support the maximal 
SE, in a prior work~\cite{queiroz-cost-ixs-19} we present the spectro-computational 
efficiency~(SCE) analysis.  We define the SC throughput of an $N$-subcarrier mapper 
as the ratio $m(N)/T(N)$ (in bits per computational instructions\footnote{or seconds, 
given the time each instruction takes in a particular computational apparatus e.g. FPGA, 
ASIC.}), where $T(N)$ is the mapper's asymptotic complexity to map $m(N)$ bits into an
$N$-subcarrier complex OFDM symbol. From this, the largest computational complexity $T(N)$ 
must satisfy $\lim_{N\to\infty}m(N)/T(N)>0$, i.e., the SC throughput must not nullify
as the system is assigned an arbitrarily large amount of spectrum. 
Based on that, in~\cite{queiroz-wcl-19} we present  the first mapper that supports 
all $2^{\lfloor\log_2 {N\choose N/2}\rfloor}$ waveforms of OFDM-IM in the same asymptotic 
time of the classic OFDM mapper. However, that proposed mapper still requires an extra 
space of $\Theta(N^2)$ look-up table entries in comparison to the classic OFDM mapper. 

\subsection{Our Contribution}
In this work, we build upon~\cite{queiroz-cost-ixs-19} and~\cite{queiroz-wcl-19}
to demonstrate the first asymptotically optimal OFDM-IM mapper. By optimal, we mean
our mapper enables all $2^{\lfloor\log_2 {N\choose N/2}\rfloor}$ waveforms of OFDM-IM
under the same asymptotic time and space complexities of the classic OFDM mapper.
Thus, we enhance our prior work~\cite{queiroz-wcl-19} by reducing the space 
complexity of the mapper from $\Theta(N^2)$ to $\Theta(N)$. Besides, we enhance
the upper-bound analysis of~\cite{queiroz-cost-ixs-19} by also showing the corresponding
asymptotic lower-bounds that holds for any OFDM-IM implementation.
In summary, we achieve the following contributions:
\begin{itemize}
  \item We derive the general OFDM-IM mapper lower-bound $\Omega(k\log_2 M + \log_2{N \choose k} + k)$
        and show it becomes the same of the classic OFDM mapper under the optimal
         configuration (i.e., $g=1$, $k=N/2$, $M=2$). \emph{This formally proves that
        enabling all OFDM-IM waveforms is not computationally intractable, as previously 
        conjectured~\cite{lu-compressiveim-2018, basar-thesurvey-2017}}; 
    \item  Based on the upper and lower bound we identify, we show that
        the optimal OFDM-IM mapper must run in exact $\Theta(N)$ asymptotic complexity.
             An implementation running above this complexity (i.e. $T(N)=\omega(N)$) nullifies the 
            SC throughput for arbitrarily large $N$, whereas one running below that
            (i.e., $T(N)=o(N)$) prevents the SE maximization;
   \item We present the first worst-case computational complexity analysis of the original 
         OFDM-IM (de)mapper when the maximal SE is allowed. In this context, we show that 
         the OFDM-IM mapper/demapper runs in $O(N^2)$ and becomes more complex than the 
         Inverse Discrete fast Fourier Transform (IDFT)/DFT algorithm;
   \item We present an OFDM-IM mapper that runs in $\Theta(N)$ time;
         % With this, we demonstrate that \emph{the maximum SE gain of $1.5\times$ over OFDM can 
         %be achieved under the same time and space asymptotic complexities of the OFDM mapper}.
   \item We implement an open-source library that supports all steps to map/demap an
         $N$-subcarrier complex frequency-domain OFDM-IM symbol. In our library,
         the IxS block is implemented with C++ \emph{callbacks}
         to enable flexible addition of other unranking/ranking 
         algorithms in the mapper. This facilitates the enhancement of currently supported algorithms
         to consider aspects not studied  in this work, e.g.~equiprobable IM waveforms~\cite{wen-equiprobable-2016},
         Hamming distance minimization~\cite{yoon-hammingmapper-2019}.
         Based on our theoretical findings, our OFDM-IM mapper library is the first implementation
         that enables the OFDM-IM SE maximization while consuming the same time and space asymptotic   
          complexities of the classic OFDM mapper.
\end{itemize}

\subsection{Organization of Work}
The remainder of this work is organized as follows.
In Section~\ref{sec:model}, we present the system model 
and the assumptions of our work. In Section~\ref{sec:scaling},
we present the computational complexity scaling laws of the OFDM-IM
mapper, namely, the lower and upper CC bounds under maximal SE.
In Section~\ref{sec:throughput}, we analyze the throughput of
the original OFDM-IM mapper. Because such analysis requires the
IxS complexity, in that section we also analyze the CC of the 
original IxS algorithm and show how to achieve
the lowest possible CC under the maximal SE.
In Section~\ref{sec:practical}, we present a practical case
study to validate our theoretical findings. Finally, in 
Section~\ref{sec:conclusion}, we conclude our work and
point future directions.

\section{System Model and Assumptions}\label{sec:model}
In this section, we review the OFDM-IM mapper (subsection~\ref{subsec:ofdmim})
and present its required design for SE maximization (subsection~\ref{subsec:optimalmapper}).
In subsection~\ref{subsec:asymptotic}, we present the assumptions to determine the
lower and upper bound complexities for the OFDM-IM mapper.

\begin{figure*}[t]
    \begin{subfigure}{.55\textwidth}
     \centering
      \includegraphics[width=\linewidth]{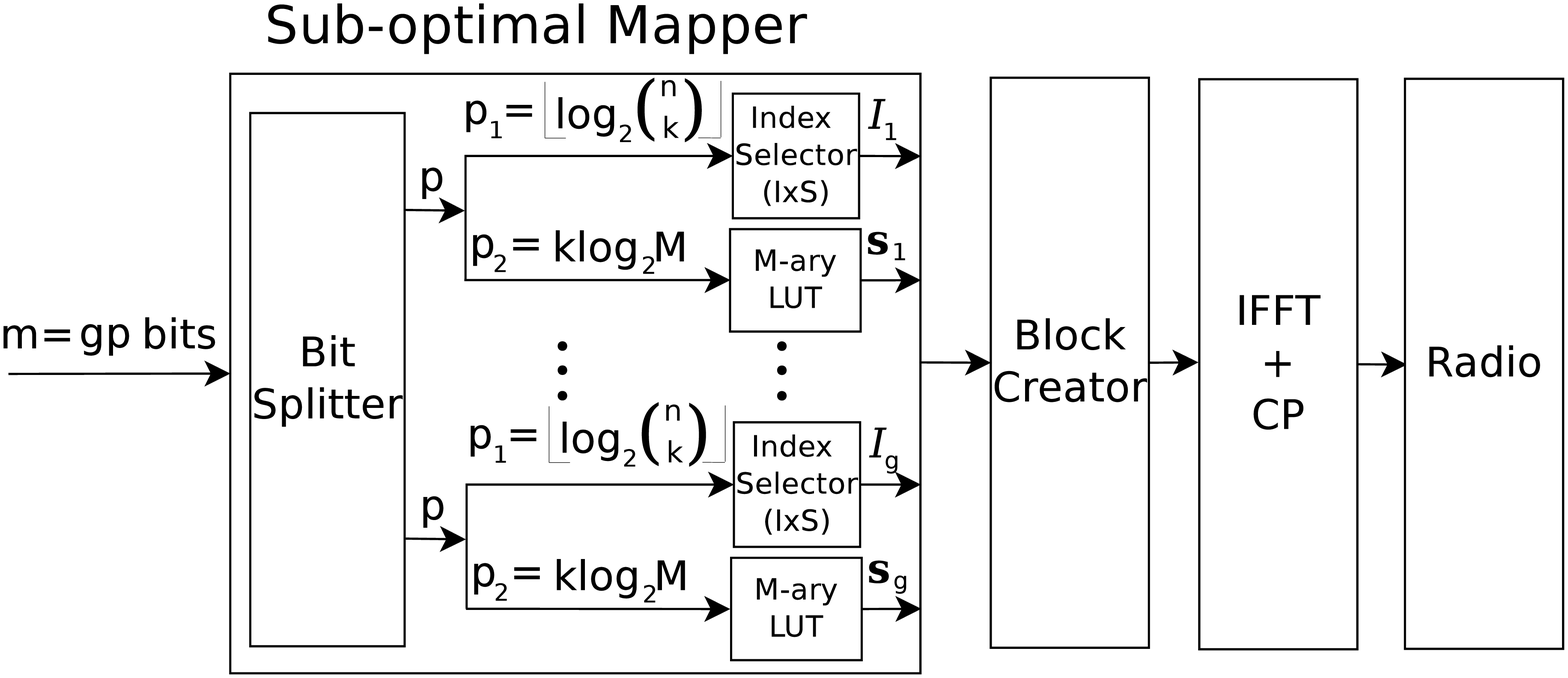}
      \caption{OFDM-IM Waveform.\label{fig:ofdmimmapper}}
    \end{subfigure}
\hfill
    \begin{subfigure}{.4\textwidth}
     \centering
      \includegraphics[width=\linewidth]{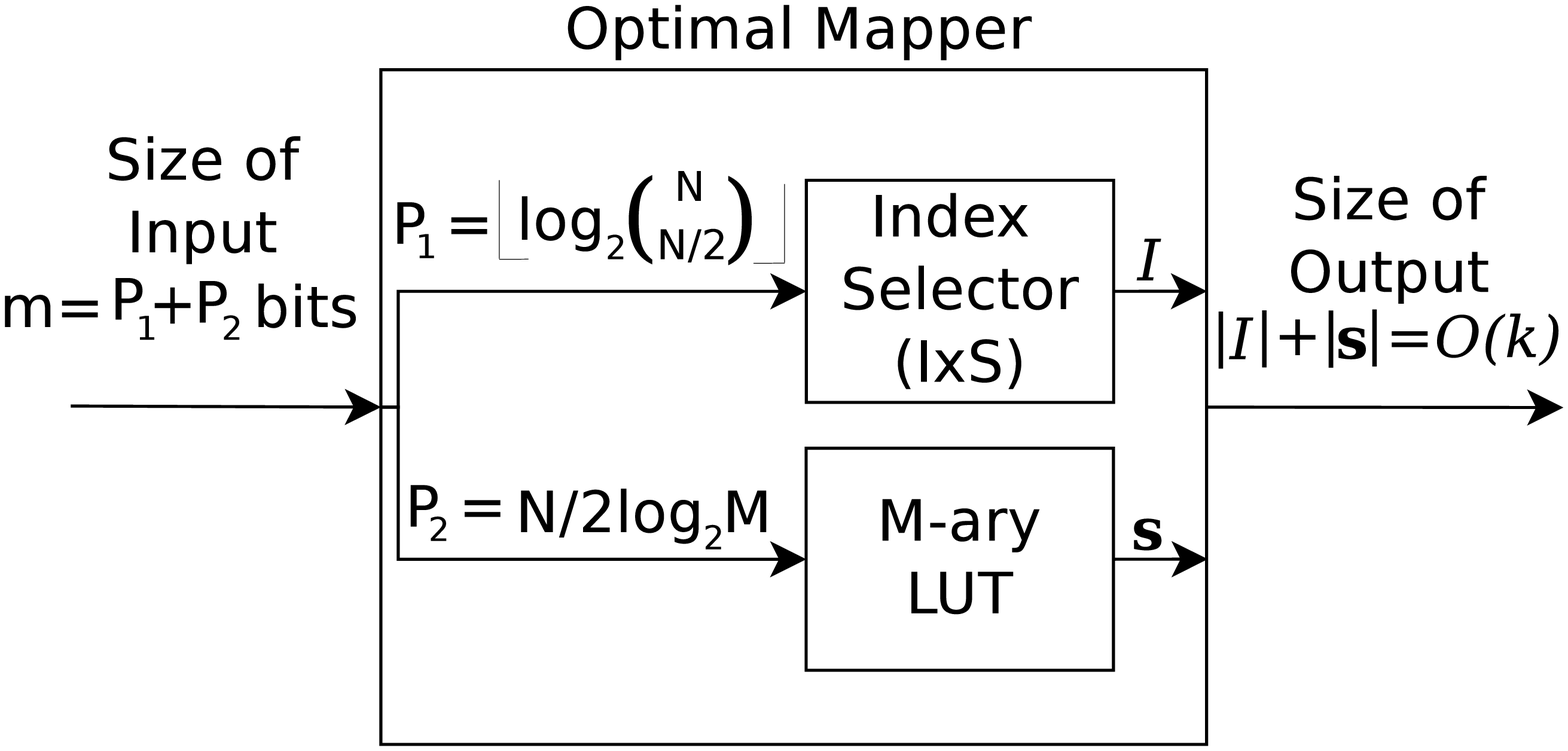}
      \caption{Optimal Mapper: $k=N/2$ and $g=1$.\label{fig:optimalmapper}}
    \end{subfigure}
      \caption{The OFDM-IM block diagram (Fig.~\ref{fig:ofdmimmapper}) mitigates
the mapping computational complexity by subdividing the symbol into $g$ small
subblocks. To maxizimize the spectral efficiency (SE) gain over OFDM, the mapper
has to set $g=1$ and $k=N/2$ (Fig.~\ref{fig:optimalmapper}).
We prove such optimal mapper can be implemented under the same time and space
asymptotic complexities of the classic OFDM mapper.}
\end{figure*}

\subsection{OFDM-IM Background}~\label{subsec:ofdmim}
The SP mapping approach~\cite{basar-ofdmim-globecom-2012, basar-ofdmim-transac-2013} 
is responsible for the main changes OFDM-IM causes to the classic OFDM transmitter block 
diagram (as illustrated in Fig.~\ref{fig:ofdmimmapper}). SP is characterized by
the configuration parameter $g\geq 1$, which stands for
the number of subblocks within the $N$-subcarrier OFDM-IM symbol. 
Each subblock has $n=\lfloor N/g\rfloor$ subcarriers out of which $k$ must 
be active. Considering an $M$-point modulator for the active subcarriers, 
each subblock maps  $p=p_1+p_2=k\log_2 M + \lfloor \log_2 {n\choose k}\rfloor$ bits 
and the entire symbol has $gp$ bits. The IxS algorithm of the $\beta$-th
subblock ($\beta=1, \dots, g$) is fed with $p_1=\lfloor \log_2 {n\choose k}\rfloor$ 
bits and outputs vector $I_\beta$, the $k$-size vector containing the indexes 
of the subcarriers that must be active in the $\beta$-th subblock. 
To modulate the $k$ active subcarriers, the ``$M$-ary modulator'' step takes
the remainder $p_2 = k\log_2 M$ bits as input and outputs 
the vector \textbf{s}$_\beta$, which consists of $k$ complex baseband signals taken 
from an $M$ constellation diagram. Then, each subblock forwards $2k$ values 
(i.e., $|$\textbf{s}$_\beta|+|I_\beta|$ ) to the ``OFDM block creator'', which
refers to \textbf{s}$_\beta$ and $I_\beta$ to modulate the $k$ active 
subcarriers in each subblock and build the full $N$-subcarrier frequency domain OFDM-IM symbol.
The remaining steps proceed as usual in OFDM~\cite{proakis2008digital}.

\subsection{Optimal OFDM-IM Mapper Design}~\label{subsec:optimalmapper}
A requisite to maximize the OFDM-IM SE is to deactivate SP 
(i.e., set $g$ to $1$) and $k$ to $N/2$~\cite{basar-ofdmim-transac-2013}.
In theory, achieving the maximal SE is just a matter of setting OFDM-IM
with the proper parameters. Indeed, by setting $g$ to $1$ (i.e., deactivating
SP) and $k$ to $N/2$, the resulting mapper (Fig.~\ref{fig:optimalmapper})
enables all $2^{P_1}$ waveforms of OFDM-IM~\cite{basar-ofdmim-transac-2013}. 
However, the authors of the original OFDM-IM waveform recommend avoiding
the ideal setup because of the resulting computational complexity (compared 
with the classic OFDM mapper).  In fact, by looking at Fig.~\ref{fig:optimalmapper}, 
one may observe that the ideal OFDM-IM mapper can be seen as a classic OFDM mapper 
with the addition of the IxS step. Because of this extra-step, the optimal
OFDM-IM mapper requires more computational steps than its OFDM counterpart.
However, our rationale is that, \emph{if one can design an OFDM-IM mapper under 
the same asymptotic computational complexity of the classic OFDM mapper, then 
the extra computational operations required by the OFDM-IM mapper (compared 
to OFDM's) are bounded by a constant even for arbitrarily large $N$}. 
Since the IxS complexity is not affected by $M$, without loss of generality, 
in this work we adopt $M=2$ to achieve the largest gain in comparison to the OFDM 
counterpart~\cite{fan-ofdm_gim1-globecom-2013,fan-ofdmgim3-2015}.
We refer to this as the optimal OFDM-IM setup.% (Def.~\ref{def:optimalSE}).
%Besides setting $g=1$ and $k=N/2$, the maximal SE gain of OFDM-IM over OFDM demands
%$M=2$. 
% Section~\ref{sec:scaling} we derive both the lower and 
%upper time complexity bounds for any OFDM-IM mapper implementation.

%\begin{definition}[Optimal Spectral Efficiency Setup]~\label{def:optimalSE}
%We define the optimal spectral efficiency OFDM-IM configuration by setting
%the number of subblocks $g$ to $1$, the number of active subcarriers
%$k$ to $N/2$ and the constellation size $M$ to 
%$2$~\cite{fan-ofdm_gim1-globecom-2013,fan-ofdmgim3-2015}. The total number
%of $N$ of subcarriers is assumed as even, as usual in OFDM systems.
%\end{definition}

\subsection{Asymptotic Analysis of Multicarrier Mappers}~\label{subsec:asymptotic}
We study the scaling laws of the OFDM-IM mapper
as a function of the number $N$ of subcarriers. In particular, for an
$N$-subcarrier OFDM-IM symbol, we study the number $m(N)$ of bits per 
symbol and the mapper's computational complexity $T(N)$ to map these 
bits into $N$ complex baseband samples. We concern about the minimum and 
maximum  asymptotic number of computational instructions required by 
any OFDM-IM mapper implementation.
For this end, we employ the asymptotic notation as usual in the 
analysis of algorithms~\cite{cormen-2009}.
 %i.e., $O()$, $o()$, $\Omega()$, $\omega()$ and  $\Theta()$~\cite{cormen-2009}. 
Our asymptotic analysis assumes the classic Random-Access Machine (RAM) model 
which is shown to be equivalent to the universal Turing
 machine~\cite{cook-rammodel-1972}. The RAM model focus on counting the
amount of basic computational instructions (e.g., data reading, data writing, 
basic arithmetic, data comparison) regardless of the technology of the
underlying computational apparatus. For example, based on the RAM model,
one verifies that a classic $N$-subcarrier BPSK-modulated OFDM mapper needs to 
perform $N$ basic computational instructions of data reading, 
each as wide as $\log_2 2$ bits.
This imposes a minimum of $\Omega(N)$ basic reading operations, regardless
of a serial or parallel implementation. Of course, performing these instructions 
in parallel yields more efficient runtime than performing them on a single processor. 
Anyway, the resources consumed by the parallel solution must scale on the 
derived computational complexity. Besides, for each reading, $N$ independent 
baseband samples must feed $N$ variables in the input of the IDFT DSP block, 
demanding a minimum space of $\Omega(N)$ complex variables.
%%~\footnote{the CC-theoretic Church-Turing thesis 
%%holding for nowadays computational apparatus can be surpassed by its ``quantum theoretic'' 
%%counterpart~\cite{kaye-quantumcomp-2007}.} 

\section{Index Modulation Mapping Complexity Bounds}~\label{sec:scaling}
In this section, we derive the CC lower and upper 
bounds for an OFDM-IM mapper implementation through asymptotic 
analysis as a function of the number of subcarriers $N$.

\subsection{OFDM-IM Mapping Time Complexity Lower Bound}
To derive the general asymptotic lower bound for any
 OFDM-IM implementation, we refer to Fig.~\ref{fig:optimalmapper}. 
Recall we are considering an SP-free mapper design (i.e., $g=1$) to enable 
the IM principle on the entire $N$-subcarrier OFDM-IM symbol. In this case, 
the lower bound is readily derived by observing that any implementation
needs at least $m$ basic computational steps to read the binary input 
to be mapped. Also,  $O(k)$ basic computational steps are required to write 
the baseband samples in the mapper's output. Based on this, 
in Lemma~\ref{lem:generallowerbound} we derive the general CC lower 
bound for any OFDM-IM mapper implementation.

\begin{lemma}[OFDM-IM Mapper General CC Lower Bound]~\label{lem:generallowerbound}
The minimum number of computational steps of any OFDM-IM mapper implementation
is $\Omega(k\log_2 M + \lfloor \log_2 {N \choose k}\rfloor + k)$.
\end{lemma}
\begin{proof}
In the optimal OFDM-IM mapper, $g=1$. Thus, the minimum number of computational 
steps to read the input is $m=P_1+P_2=\lfloor \log_2 {N \choose k}\rfloor+k\log_2 M$.
Further, the OFDM-IM mapper must feed the ``OFDM block creator'' DSP step with
the vectors of the active subcarriers indexes $I_\beta$ and their corresponding
baseband samples \textbf{s}$_\beta$ ($\beta=1, \dots, g$). 
Since the optimal mapper requires $g=1$, 
there is only a single $k$-size vector $I_1$ and another $k$-size vector
\textbf{s}$_1$, yielding to the total output size of $2k=O(k)$. Thus, any OFDM-IM
mapper implementation must write at least $O(k)$ units of data in its output. 
Therefore, because of the computational effort to read (input) and write (output) 
units of data, any OFDM-IM mapper solution will demand at least
$\Omega(k\log_2 M + \lfloor \log_2 {N \choose k}\rfloor + k)$ computational
steps.
\end{proof}

When the optimal OFDM-IM setup is allowed, the general asymptotic lower bound of 
Lemma~\ref{lem:generallowerbound} becomes $\Omega(N)$ (Corollary~\ref{col:optimallowerbound}).
This stems from the fact that the number of index modulated bits $P_1$ approaches
$N-\log_2\sqrt N$ as $N\to\infty$ (Lemma~\ref{lemma:p_1}).
Therefore, although the number of waveforms of the optimal OFDM-IM setup grows 
exponentially on $N$, the CC of the IM mapping problem is not intractable 
(i.e., $\Omega(2^N)$) as previously conjectured~\cite{lu-compressiveim-2018,basar-thesurvey-2017}.

\begin{lemma}[Maximum Number $P_1$ of Index Modulation Bits] \label{lemma:p_1}
The maximum number of index modulated bits $P_1$ approaches 
$N-\log_2\sqrt N$ for arbitrarily large $N$.
\end{lemma}
\begin{proof}
By definition, $P_1 = \lfloor\log_2 {N\choose k}\rfloor$.
If the maximum SE gain of OFDM-IM over OFDM is allowed, ${N\choose k}$
becomes the so-called central binomial coefficient 
${N\choose N/2}$, whose well-known asymptotic growth 
is $O(2^N/\sqrt{N})$\cite{oeis-A001405}.  From this, it follows
that $P_1$ approaches $\log_2 ({2^NN^{-0.5}})=N-\log_2\sqrt N=O(N)$ as $N\to\infty$.
\end{proof}

\begin{corollary}[OFDM-IM Mapper CC Lower Bound under 
Maximal Spectral Efficiency] \label{col:optimallowerbound}
Under the optimal spectral efficiency setup, the general mapping CC 
lower bound of OFDM-IM (Lemma~\ref{lem:generallowerbound}) becomes $\Omega(N+P_1)$, 
which is the same of OFDM, i.e., $\Omega(N)$.
\end{corollary}
\begin{proof}
Since $P_1$ approaches $N-\log_2\sqrt{N}=O(N)$ for arbitrarily large 
$N$ (Lemma~\ref{lemma:p_1}), the general asymptotic lower-bound $\Omega(N+P_1)$
 becomes $\Omega(N)$, which is the minimum asymptotic number of computational steps 
performed by the classic OFDM mapper.
\end{proof}

Lemma~\ref{lem:generallowerbound} 
and Corollary~\ref{col:optimallowerbound} imply that it is not possible to 
implement an OFDM-IM mapper with less than $\Omega(N)$ computational steps 
without sacrificing the SE optimality 
(Corollary~\ref{col:lowerboundtradeoff}).  The corollary~\ref{col:lowerboundtradeoff} 
states that any OFDM-IM mapper running in sub-linear complexity, i.e., $k=o(N)$ (which 
excludes the ideal $k=N/2$), prevents the maximal SE gain over OFDM.
%This is illustrated in  Fig.~\ref{fig:differentk}, in which we plot the OFDM-IM SE gain against 
%OFDM for different asymptotic choices of $k$ as function of $N$.
However, sub-optimal SE setups can be useful for sparse OFDM-IM 
systems, in which one gives up the maximal throughput on behalf 
of energy consumption minimization~\cite{salah-2019}.

\begin{corollary}[OFDM-IM Mapper Spectro-Computational Lower-Bound 
Trade-Off]~\label{col:lowerboundtradeoff}
No OFDM-IM mapper implementation can maximize the spectral efficiency (SE)
gain over OFDM while running in $o(N)$ computational steps.
\end{corollary}
\begin{proof}
The asymptotic number of steps of any OFDM-IM mapper is subject to the general
lower bound of $\Omega(k\log_2 M + \lfloor \log_2 {N \choose k}\rfloor + k)$ 
(Lemma~\ref{lem:generallowerbound}).
Thus, the only way to improve that bound consists of changing the OFDM-IM configuration 
parameters $M$ and $k$. Out of all possible values of $M$ and $k$, the \emph{maximum} SE 
gain of OFDM-IM over OFDM is achieved \emph{only} when $M=2$ and 
$k=N/2$~\cite{fan-ofdm_gim1-globecom-2013,fan-ofdmgim3-2015}. Also, under 
such optimal SE configuration, the general  CC lower bound becomes 
$\Omega(N)$~(Corollary~\ref{col:optimallowerbound}). 
Therefore, an OFDM-IM implementation cannot run bellow this bound (i.e., in sub-linear time)
unless a non-optimal SE configuration is adopted for $k$.
\end{proof}

\begin{comment}
\begin{figure}[t]
  \includegraphics[width=\linewidth]{graphics/se-gain/theoretical/limit.eps}
  \caption{Maximum theoretical Spectral Efficiency (SE) gain of OFDM-IM over OFDM for different choices of $k$ and
arbitrarily large $N$. Under the optimal setting ($k=N/2$) the gain limit is $1.5\times$ the OFDM SE.}
  \label{fig:differentk}
\end{figure}
\end{comment}

\begin{comment}
\begin{table}
\caption{OFDM-IM Mapping CC Lower Bounds Across
Different Choices of $k$.}% (Theorem~\ref{th:ixsn}).}
\label{tb:lowerboundsummary}
\centering
\begin{tabular}{|l|l|c|}
\hline
\multicolumn{2}{|c|}{\begin{tabular}[c]{@{}c@{}}Asymptotic formula\\ for $k$\end{tabular}} & \multirow{2}{*}{\begin{tabular}[c]{@{}c@{}} OFDM-IM Mapping Lower Bound \end{tabular}} \\ \cline{1-2}
Formula & \multicolumn{1}{c|}{Example} &  \\ \hline \hline
$\Theta(1)$ & $k=4$ & $\Omega(\log_2$) \\ \hline
$\Theta(\sqrt{N})$ & $k=\lfloor\sqrt{N}\rfloor$ & $\Omega(\sqrt{N}\log_2 N)$ \\ \hline
$\Theta(\log{N})$ & $k=\lfloor\log_2{N}\rfloor$ & $\Omega(??)$ \\ \hline
$\Theta(N)$ & $k=\lfloor N/2\rfloor$ & $\Omega(N)$ \\ \hline
\end{tabular}
\end{table}
\end{comment}

\subsection{OFDM-IM Mapping Time Complexity Upper Bound}
The CC upper bound of a problem is usually defined as the 
complexity of the fastest currently known algorithm that
solves it~\cite{harel-1987}. 
This definition does not suffice to our study because our asymptotic analysis
is further constrained by the SE maximization. In fact, if the fastest known 
algorithm does not suffice to avoid an increasing bottleneck in the mapping 
throughput as $N$ grows, then its complexity 
cannot be considered suitable to scale the mapper throughput on $N$.
From this, we define the spectro-computational mapper throughput 
(Def.~\ref{def:sct}) and, based on its condition of scalability 
(Def.~\ref{def:sca}),  we derive the required computational 
complexity upper bound for any OFDM-IM mapper implementation (Lemma~\ref{lemma:upperbound}).

\begin{definition}[The Spectro-Computational (SC) Throughput]\label{def:sct}
Let $T(N)$ be the computational complexity (CC) to map $m(N)$ input bits into 
an $N$-subcarrier OFDM-IM symbol. We define $m(N)/T(N)$ in bits per 
computational steps (or seconds), as the spectro-computational (SC) throughput 
of the mapper.
\end{definition}

\begin{definition}[Spectro-Computational Throughput Scalability]\label{def:sca}
The SC throughput $m(N)/T(N)$ of a mapper is not scalable unless the 
inequality~(\ref{eq:condition}) does hold.
\begin{eqnarray}
\lim_{N\to\infty} \frac{m(N)}{T(N)} &>& 0 \label{eq:condition}
\end{eqnarray}
\end{definition}

As a side note about our Def.~\ref{def:sca}, we call attention to the
fact that it consists of the asymptotic analysis. As such, ``time complexity''
means ``amount of computational instructions'' which can be translated to
(but does not necessarily mean) wall clock runtime. That said, we recognize 
that a radio implementation that does not meet our Def.~\ref{def:sca} 
can achieve the same wall clock runtime of another one that does.
However, in this case, the CC $T(N)$ will translate into other relevant 
radio's design performance  indicators. For example, suppose that the 
largest complexity $T(N)$ to satisfy our Def.~\ref{def:sca} in a particular 
DSP study is $O(N)$. A design that violates such a requirement by employing a 
more complex algorithm, let us say $O(N^2)$, can still reach the same wall 
clock runtime of a design that does not. However, since the overall number
of performed computational instructions depends on the algorithm's CC rather 
than the hardware technology, the average wall clock time to run  a single 
computational instruction must be (much) lower in the $O(N^2)$ solution in 
comparison to the $O(N)$ counterpart. This pushes the algorithm's CC to the 
hardware design rather than to the wall clock runtime. Therefore, the SC throughput  
of a radio design that violates our Def.~\ref{def:sca} can scale with $N$ but 
at the expense of impairing other relevant design performance indicators, such 
as the number of hardware components (e.g., logic gates), circuit area, 
energy consumption and manufacturing cost~\cite{blume-socalgorithmdesign-2002}. 

\subsection{Required Complexity for Maximal SE}
Based on~Def.~\ref{def:sca}, in Lemma~\ref{lemma:upperbound}
we show that the upper bound complexity any OFDM-IM mapper 
implementation  must meet to ensure the optimal SE configuration is $O(N)$.

\begin{lemma}[OFDM-IM Mapper Upper Bound under Optimal SE Configuration]\label{lemma:upperbound}
Under the optimal SE configuration, the OFDM-IM mapper
CC must be upper bounded by $O(N)$.
\end{lemma}
\begin{proof}
To meet the inequality~\ref{eq:condition} of Def.~\ref{def:sca},
$T(N)$ must be asymptotically less or equal than $m(N)$, i.e.,
$T(N)=O(m(N))=O(P_1+P_2)$. Under the optimal SE configuration, 
$k=N/2$ and $P_1=\log_2{N\choose N/2}=O(N)$ bits (Lemma~\ref{lemma:p_1}). 
Therefore, $T(N)$ must be $O(N)$.
\end{proof}

Based on the fact that the required OFDM-IM mapper upper bound complexity
matches its lower bound order of growth in the optimal SE configuration, 
Theorem~\ref{th:imorder} tells us that \emph{the OFDM-IM mapper must run in 
$\Theta(N)$ time}. A solution requiring more asymptotic steps (i.e., $\omega(N)$) 
nullifies the mapper throughput as $N$ grows, whereas one requiring fewer steps 
(i.e., $o(N)$) prevents the SE gain maximization (Corollary~\ref{col:lowerboundtradeoff}).

\begin{theorem}[Required OFDM-IM Mapping Complexity]\label{th:imorder}
If the configuration that maximizes the OFDM-IM spectral efficiency gain
over OFDM is allowed (i.e., $g=1$, $k=N/2$, $M=2$), 
the OFDM-IM mapper block of~\cite{basar-ofdmim-globecom-2012, 
basar-ofdmim-transac-2013} must run in $\Theta(N)$ computational steps.
\end{theorem}
\begin{proof}
Corollaries~\ref{col:optimallowerbound} and~\ref{col:lowerboundtradeoff} 
show that any OFDM-IM mapper implementation running with less than $\Omega(N)$ 
computational steps cannot achieve the optimal SE gain over OFDM. In turn,
Lemma~\ref{lemma:upperbound} tells us that the mapper throughput nullifies
for arbitrarily large $N$ if its complexity requires more than $O(N)$ steps.
Therefore, the exact asymptotic number of computational steps for any OFDM-IM 
mapper implementation under the optimal SE configuration must be $\Theta(N)$.
\end{proof}

%~\cite{blume-socalgorithmdesign-2002}

\section{Throughput Analysis}\label{sec:throughput}
Our theoretical findings summarized in Theorem~\ref{th:imorder},
disclose the conditions for the computational feasibility of the
optimal OFDM-IM mapper. The theorem requires exactly $\Theta(N)$ steps
for the mapper. Since the $M$-ary LUT block of the 
OFDM-IM mapper (Fig.~\ref{fig:optimalmapper}) already runs in $N/2=O(N)$ 
computational steps, to meet the theorem we just need to demonstrate
the IxS block can be implemented with $\Theta(N)$ computational steps.

By relying on the literature in combinatorics, one can achieve (un)ranking
complexities faster than the $\Theta(N)$ time required by our 
Theorem~\ref{th:imorder} e.g.,~\cite{parque-2018, shimizu-unrakingsmalk-2014}. 
Such a performance, however, demands $k=o(N)$. Translated to the OFDM-IM domain, 
this means such algorithms prevent the SE maximization 
(Corollary~\ref{col:lowerboundtradeoff}). We identify that the original OFDM-IM 
mapper (and its variants) refer to the (un)ranking algorithm named 
``Combinadic''~\cite{buckles-comb-1977,combinadic-2014}\footnote{In \cite{crouse-2007},
the author points a fix to the algorithm of~\cite{buckles-comb-1977}.}.
In Subsection~\ref{subsec:combinadic}, we analyze the OFDM-IM SCE
having Combinadic as the IxS block. We show that the Combinadic algorithm 
not only prevents the mapper to meet our Theorem~\ref{th:imorder} but also 
surpasses the $O(N\log_2 N)$ complexity of the IDFT DSP algorithm. 
In Subsection~\ref{subsec:optimal}, we propose an optimal OFDM-IM mapper by
adapting Combinadic to run in linear rather than  quadratic complexity.

\subsection{OFDM-IM Mapper with Combinadic}~\label{subsec:combinadic}
We start this subsection by explaining how the Combinadic algorithm
works. Then, we analyze its CC when the optimal SE configuration of OFDM-IM is allowed.
Based on that, we conduct the spectro-computational analysis of the OFDM-IM mapper.

\subsubsection{Combinadic Terminology}
The Combinadic algorithm relies on the fact that each decimal number $X$ in the integer range 
$[0,{N\choose k}-1]$ has an unique representation $(c_k,\cdots,c_2,c_1)$ in the combinatorial number 
system\cite{knuth-art-combinatorial-2011}~(Eq.~\ref{eq:base}). 
For OFDM-IM, $X$ represents the $P_1$-bit input (in base-10) and the coefficients 
$c_k>\cdots>c_{2}>c_1\geq 0$ represent the indexes of the $k$ subcarriers that must be 
active in the subblock.
\begin{eqnarray}\label{eq:base}
X &=& {c_{k} \choose k} + \cdots + {c_{2} \choose 2} + {c_{1} \choose 1}
\end{eqnarray}

Combinadic may refer to two distinct tasks, namely, unranking and ranking.
The Combinadic unranking (shown in Alg.~\ref{alg:unrankingcombinadic})
consists in computing the array of coefficients $c_i$, $i\in[1,k]$, of Eq.~(\ref{eq:base})
from the input $X$ (along with $N$ and $k$). The Combinadic unranking
takes place in the IxS of the OFDM-IM transmitter. The reverse process, i.e.,
computing $X$ given all $k$ coefficients $c_i$, $i\in[1,k]$, is known as ranking
and is performed by the IxS of the OFDM-IM receiver (Alg.~\ref{alg:rankingcombinadic}).

\subsubsection{Combinadic Unranking Functioning}
%Considering that $N$,
%$k$ and $X$ are input parameters, the algorithm aims to find out the value for each
%coefficient $c_i$, $i\in[1,k]$ such that the Eq.~\ref{eq:base} holds. 
The Combinadic unranking is shown in Alg.~\ref{alg:unrankingcombinadic}. 
It takes $N$, $k$ and $X$ as input parameters and outputs the array $c_i$,
$i\in[1,k]$ such that $X =\sum_{i=1}^{k}{c_i\choose i}$ (Eq.~\ref{eq:base}).
The \emph{candidate} values for the coefficients $c_i$ considered by the 
algorithm are $0,1,\cdots,N-1$, which represent the indexes of the $N$ subcarriers. 
The coefficients are determined from $c_k$ until $c_1$ and  the variable $cc$ 
(line~\ref{ln:unrankingcc}) stores the next candidate value for the 
current coefficient being computed. The first coefficient to be computed is $c_k$
and its first candidate is $N-1$. This is the value of $cc$ in the very first execution
of line~\ref{ln:unrankingn1}. 
For every candidate value $cc$, the  corresponding binomial coefficient  ${cc\choose i}$
is computed and stored in the variable $ccBinCoef$ (line~\ref{ln:unrankingbc}).
If condition $ccBinCoef\leq X$ is satisfied (line~\ref{ln:unrankinginnercond}), then 
the candidate value $cc$ is confirmed as the value of $c_i$ (line~\ref{ln:unrankingcfound})
and $X$ is updated accordingly (line~\ref{ln:unrankingupdatex}). This entire process 
repeats until all the remainder $k-1$ coefficients are determined.

%\begin{figure*}[ttt!]
\begin{procedure*}[ttt!]
\begin{minipage}{0.46\textwidth}
    \begin{algorithm}[H]
    \begin{algorithmic}[1]
    \STATE \COMMENT{Inputs: $X$, $N$, and $k\in [1,N]$}
    \STATE \COMMENT{Output: Array $c_i$ ($i\in[1,k]$) such that $X =\sum_{i=1}^{k}{c_i\choose i}$ (Eq.~\ref{eq:base})}
%\Statex
        \STATE $cc \leftarrow N$;\label{ln:unrankingcc}\COMMENT{the current next candidate for $c_i$};
        \FOR{$i$ \textbf{from} $k$ \textbf{downto} $1$}
         %\STATE $c_i \leftarrow largestCandidate$; 
        \REPEAT 
            \STATE{$cc \leftarrow cc - 1;$}\label{ln:unrankingn1} \COMMENT{the first candidate for $c_k$ is $N-1$};
            \STATE{$ccBinCoef \leftarrow {cc \choose i};$} \label{ln:unrankingbc}
        \UNTIL{$ccBinCoef\leq X$} \label{ln:unrankinginnercond}
         \STATE $c_i \leftarrow cc$; \label{ln:unrankingcfound}
         \STATE $X \leftarrow X - ccBinCoef$;\label{ln:unrankingupdatex}
        \ENDFOR
        \RETURN  array $c$;
    \end{algorithmic}
      \caption{Combinadic Unranking (OFDM-IM IxS Transmitter).}\label{alg:unrankingcombinadic}
    \end{algorithm}
\end{minipage}
\hfill
\begin{minipage}{0.46\textwidth}
    \begin{algorithm}[H]
    \begin{algorithmic}[1]
   \STATE \COMMENT{Inputs: Array $c_k>\cdots>c_2>c_1\geq 0$, $N>c_k$, and $k\in [1,N]$}
    \STATE \COMMENT{Output: {$X =\sum_{i=1}^{k}{c_i\choose i}$} (Eq.~\ref{eq:base})};
    \STATE $X \leftarrow 0$;
    \FOR{$i$ \textbf{from} $1$ \TO $k$}\label{ln:ranking1loop}
     \STATE $X \leftarrow X + {c_i\choose i}$;
    \ENDFOR
    \RETURN $X$;
    \end{algorithmic}
      \caption{Combinadic Ranking (OFDM-IM IxS Receiver).
              \label{alg:rankingcombinadic}}
    \end{algorithm}
\end{minipage}
%\captionsetup{labelformat=simple}
      \caption*{Combinadic unranking and ranking algorithms referred to by the IxS block of original OFDM-IM mapper. In the maximal spectral efficiency OFDM-IM mapper (Fig.~\ref{fig:optimalmapper}), these algorithms run in $O(N^2)$, surpassing the computational complexity of the Fourier transform algorithm. \label{fig:combinadic}}
\end{procedure*}

\subsubsection{Combinadic Unranking Complexity}
In a particular \emph{worst-case} instance of Combinadic unranking (Alg.~\ref{alg:unrankingcombinadic}), 
the logic test of the inner loop
(line~\ref{ln:unrankinginnercond}) fails for $cc=N-1, N-2, \cdots, k$
in the first iteration of the outer loop, i.e. when the first coefficient 
$c_k$ is being determined. 
Thus, $c_k$ is assigned to $k-1$. This narrows the list of candidates
(for the remainder $k-1$ coefficients) to the values $k-2,k-3,\cdots,1,0$ .
Since the combinatorial number system ensures that all $k$ coefficients are distinct and  
that $c_k$ is the largest one, a candidate value that fails for $c_k$ can be discarded
for $c_{k-1}$ and so on. Thus, after $c_k$ is determined, there must be at least $k-1$ 
candidate values for the remainder $k-1$ coefficients. Because of this, there is only
one logic test per candidate value in the inner loop regardless of the number of coefficients. 
Since there are $N$ candidate values, the inner loop takes $O(N)$ time regardless of the outer loop. 
In each test of the inner loop, Combinadic relies on the multiplicative identity 
(Eq.~\ref{eq:multiplicativeformula}) to compute the binomial coefficient value in $O(k)$ time.
\begin{eqnarray}
%{n \choose k} &=& \frac{n(n-1)(n-2)\cdots(n-(k-1))(n-k)!}{k!(n-k)!} \nonumber\\
%& =& \frac{n(n-1)(n-2)\cdots(n-(k-1))}{k!} \nonumber\\
%&=& \frac{(n-1+1)(n-2+1)(n-3+1)\cdots(n-k+1))}{k(k-1)(k-2)\cdots 1} \nonumber\\
{N \choose k} &=& \prod_{i=1}^{k}\frac{N-i+1}{i} \label{eq:multiplicativeformula}
\end{eqnarray}
Therefore, the overall CC of the Combinadic unranking algorithm
is $O(Nk)$. Considering the optimal SE configuration, $k=N/2$ and the 
complexity becomes $O(N^2)$, which is asymptotically higher than the 
$O(N\log N)$ complexity of the IDFT block.

\subsubsection{Combinadic Ranking Functioning and Complexity}
The Combinadic ranking is shown in Alg.~\ref{alg:rankingcombinadic}. 
It takes the array of coefficients $c_i$, $i\in[1,k]$ from the 
OFDM-IM detector and performs a summation of the $k$ binomial coefficients
${c_{1} \choose 1} + {c_{2} \choose 2} +\cdots + {c_{k} \choose k}$
(Eq.~\ref{eq:base}). Since each binomial coefficient ${c_i\choose i}$
can be calcuated in $O(i)$ time by the  multiplicative formula 
(Eq.~\ref{eq:multiplicativeformula}), and $i$ ranges from $1$ to $k$,
the total number of multiplications performed by the algorithm is 
$1+2+\cdots+k=k(k+1)/2=O(k^2)$. Considering the optimal OFDM-IM setup,
$k=N/2$, the overall complexity becomes $O(N^2)$ as with Combinadic 
unranking.

\subsubsection{OFDM-IM Mapper Throughput with Combinadic}\label{subsubsec:ofdmimsc}
We now analyze the SC throughput of the OFDM-IM mapper assuming the IxS block is
implemented by the Combinadic algorithm~\cite{combinadic-2014,buckles-comb-1977} 
as in the original OFDM-IM design~\cite{basar-ofdmim-transac-2013}. 
Considering the optimal OFDM-IM setup, the total
number of bits per symbol is $N/2+\lfloor\log_2{N\choose N/2}\rfloor$, whereas
the IxS complexity is $O(N^2)$, as previously analyzed. Thus, according to 
Def.~\ref{def:sca}, the resulting SC throughput must satisfy Ineq.~(\ref{eq:sc1a}) 
as follows, otherwise it nullifies over $N$. 

\begin{eqnarray}
\lim_{N\to\infty} \frac{N/2 +\lfloor\log_2{N\choose N/2}\rfloor}{O(N^2)} &\stackrel{?}{>}& 0 \label{eq:sc1a}
%\lim_{N\to\infty} \frac{N/2 + N-\log_2\sqrt N}{\kappa\cdot N^2}  &=& 0 \label{eq:sc1b}
\end{eqnarray}

According to the theory of computational complexity, 
the wall-clock time taken by a particular implementation of a $O(N^2)$ algorithm
is bounded by the function $\kappa N^2$, in which the constant $\kappa>0$ 
%is the wall-clock time taken by the execution of single instruction of the algorithm 
%in a real machine (e.g. general purpose CPU, FPGA).
captures the  wall-clock runtime taken by the asymptotic dominant 
instruction of the algorithm\footnote{The instruction we choose to count in the 
analysis. Mostly, real or complex arithmetic instructions for DSP algorithms.} 
on a real machine. In turn, the number of index modulated bits tends to
$N-\log_2\sqrt{N}$ as $N$ grows (Lemma~\ref{lemma:p_1}). With basic calculus, one can 
verify that the limit in Ineq.~(\ref{eq:sc1a}) tends to zero for arbitrarily large $N$ 
regardless of the value of $\kappa$, as follows.%(Eq.~\ref{eq:sc1b}). 
\begin{eqnarray}
%\lim_{N\to\infty} \frac{N/2 +\lfloor\log_2{N\choose N/2}\rfloor}{O(N^2)} &\stackrel{?}{>}& 0 \label{eq:sc1a}\\
\lim_{N\to\infty} \frac{N/2 + N-\log_2\sqrt N}{\kappa\cdot N^2}  &=& 0 \label{eq:sc1b}
\end{eqnarray}

Therefore, referring to the original Combinadic algorithm to implement the IxS block
in the optimal SE configuration causes the SC throughput of the OFDM-IM mapper to nullify as $N$ grows.

\subsection{Optimal Spectro-Computational Mapper}~\label{subsec:optimal}
To avoid the asymptotic nullification of the OFDM-IM mapper throughput while assuring
the maximal SE, the IxS (un)ranking algorithm must run nor faster nor slower than
$\Theta(N)$ (Thm.~\ref{th:imorder}). In~\cite{kokosinski-1995}, the author presents four 
unranking algorithms, out of which one (called ``unranking-comb-D'') can meet 
that requirement. Therefore, one can consider that algorithm to validate our
theoretical findings. However, we remark that the Combinadic algorithm
(referred to by the original OFDM-IM design) can benefit from the same properties 
of unranking-comb-D to run in $\Theta(N)$ rather than $O(N^2)$ under the optimal OFDM-IM
setup. Similarly, the ranking algorithm (not proposed in~\cite{kokosinski-1995}) can 
also run in $O(N)$ as well. Next, we explain how to adapt Combinadic to enable
the minimum possible CC when the maximal SE is allowed.

\subsubsection{Linear-time Combinadic Unranking}
The main bottleneck in the time complexity of Combinadic (un)ranking (Alg.~\ref{alg:unrankingcombinadic})
is the inner loop. As previously explained, the inner loop takes $k$ iterations, each 
of which demands further $O(i)$ iterations to compute the binomial coefficients ${c_i\choose i}$.
Since $i$ ranges from $k$ to $1$ and the optimal OFDM-IM setup imposes $k=O(N)$,
this yields $k\cdot O(i)= N/2\times O(N/2)=O(N^2)$. 
To improve this complexity, note that only the first candidate binomial 
coefficient ${c_k\choose k}={N-1 \choose N/2}$ 
needs to be computed from scratch (in $O(k)$ time). 
Thus, such computation can be performed outside both loops 
of Combinadic (Alg.~\ref{alg:unrankingcombinadic}) and stored in a variable we refer to as
$ccBinCoef$. The resulting 
modification is shown in line~\ref{ln:unrank2firstbincoef} of the Linear-time Combinadic
unranking (Alg.~\ref{alg:unranking2}).
In this algorithm, the variables $cc$ and $ccBinCoef$ denote the candidate
values for $c_i$ and ${c_i\choose i}$, respectively.
%where the variable $ccBinCoef$  is initialized in $O(N/2)$ assuming the optimal OFDM-IM configuration.
Following $ccBinCoef={c_k\choose k}$, the next candidate binomial coefficient, 
either ${N-1 \choose N/2-1}$ or ${N-2 \choose N/2-1}$, 
can be computed from $ccBinCoef$ itself in $O(1)$ time. In general,
one can calculate ${c_i-1\choose i}$ and ${c_i-1\choose i-1}$ from
${c_i\choose i}$ by relying on the following respective equations~\cite{kokosinski-1995}:
\begin{eqnarray}
{c_i-1\choose i}&=&((c_i-i)*{c_i\choose i})/c_i \label{eq:binprop1}\\
{c_i-1\choose i-1}&=&(i*{c_i\choose i})/c_i \label{eq:binprop2}
\end{eqnarray}

 The Eqs. (\ref{eq:binprop1}) and (\ref{eq:binprop2}) are exploited by lines~\ref{ln:unranking2property1} 
and~\ref{ln:unranking2property2} of Alg.~\ref{alg:unranking2}, respectively.
Thus, all remainder binomial coefficients within the logic test of the inner loop are 
computed in $O(1)$ time. Therefore, the complexity of Combinadic unranking improves 
from $k\cdot O(i)= N/2\times O(N/2)=O(N^2)$ to $O(k) + k\cdot O(1)$, yielding 
$N/2 + N/2\times O(1)=O(N)$ in the optimal OFDM-IM configuration.
\begin{procedure*}[t!]
\begin{minipage}{0.46\textwidth}
    \begin{algorithm}[H]
    \begin{algorithmic}[1]
   \STATE \COMMENT{Inputs: $X$, $N$, and $k\in [1,N]$}
    \STATE \COMMENT{Output: Array $c_i$ ($i\in[1,k]$) such that $X =\sum_{i=1}^{k}{c_i\choose i}$ (Eq.~\ref{eq:base})}
    \STATE $cc \leftarrow N-1$; \COMMENT{largest candidate for $c_i$};
    \STATE $ccBinCoef \leftarrow {cc \choose k}$; \COMMENT{candidate value for ${c_k\choose k}$};\label{ln:unrank2firstbincoef}
    \FOR{$i$ \textbf{from} $k$ \textbf{downto} $1$}
     \STATE $c_i \leftarrow cc$; 
    \WHILE{$ccBinCoef>X$}
    % \STATE $c_i \leftarrow c_i-1$;
  \STATE \COMMENT{Below, ${c_{i}-1 \choose i}$ is computed from ${c_{i}\choose i}$ in $O(1)$};
     \STATE $ccBinCoef$$\leftarrow$$((c_i$$-$$i)$*$ccBinCoef)/c_i$;~\label{ln:unranking2property1}
     \STATE $c_i \leftarrow c_i-1$;
    \ENDWHILE 
     \STATE $X \leftarrow X - ccBinCoef$; %\COMMENT{\emph{here $b={c_i \choose i}$}}
  \STATE \COMMENT{Below, ${c_{i}-1 \choose i-1}$ is computed from ${c_{i}\choose i}$ in $O(1)$};
     \STATE $cc \leftarrow c_i-1$;
      \IF{$cc = 0$} \STATE{\textbf{return} array c}
      \ENDIF
     \STATE $ccBinCoef \leftarrow (i*ccBinCoef)/c_i$;~\label{ln:unranking2property2}
    \ENDFOR
    \STATE{\textbf{return} array c}
    \end{algorithmic}
      \caption{Linear-time Combinadic Unranking (OFDM-IM Index Selector Transmitter). \label{alg:unranking2}}
    \end{algorithm}
\end{minipage}
\hfill
\begin{minipage}{0.46\textwidth}
    \begin{algorithm}[H]
    \begin{algorithmic}[1]
   \STATE \COMMENT{Inputs: Array $c_k>\cdots>c_2>c_1\geq 0$, $N>c_k$, and $k\in [1,N]$}
    \STATE \COMMENT{Output: {$X =\sum_{i=1}^{k}{c_i\choose i}$} (Eq.~\ref{eq:base})};
    \STATE $i \leftarrow 1$; 
%    \STATE \COMMENT{Next loop finds $i$ such that ${c_i\choose i}\not=0$};
    \WHILE{$i\leq k$ \AND$c_i<i$} \label{ln:ranking2specialcaseinit} 
       \STATE $ i \leftarrow i + 1$;      
    \ENDWHILE
      \IF {$i>k$} {\RETURN 0;}  %\COMMENT{if $i>k$ then $\sum_{i=1}^{k}{c_i\choose i}=0$};}
    \ENDIF   \label{ln:ranking2specialcaseend}   
    \STATE $ccBinCoef \leftarrow {c_i \choose i}$~\label{ln:ranking2initializebc}; $X \leftarrow 0$; 
    \FOR{$cc$ \textbf{from} $c_i$ \TO $N-1$} \label{ln:rank2for} 
      \IF{$c_i = cc$} \label{ln:rank2check}
        \STATE 
        {
           $X \leftarrow X + ccBinCoef$;
          \STATE $ccBinCoef$~$\leftarrow$~$(ccBinCoef$$*$$(c_i$$+$$1))$/$($$i$$+$$1)$;\label{ln:rank2improve1}%\COMMENT{$O(1)$ time.}
         \STATE $i \leftarrow i + 1$;
        }
      \ELSE 
     \STATE
       {
         $ccBinCoef$$\leftarrow$ $(ccBinCoef*(cc+1))$/$(cc+1-i)$;
       } 
      \ENDIF
   % \ENDWHILE
   \ENDFOR
   \RETURN $X$;
    \end{algorithmic}
      \caption{Linear-time Combinadic Ranking (OFDM-IM Index Selector Receiver). \label{alg:ranking2}}
    \end{algorithm}
\end{minipage}
      \caption*{Adaptation of the Combinadic algorithms 
(unranking~\ref{alg:unrankingcombinadic} and ranking~\ref{alg:rankingcombinadic})
referred to by the original OFDM-IM mapper to run in $O(N)$ time. We 
prove these adaptations enable the overall OFDM-IM mapper to maximize 
the spectral efficiency gain over OFDM while consuming the same time 
and space computational complexities of the classic OFDM mapper.
    \label{alg:lixsunranking}}
\end{procedure*}
\subsubsection{Linear-time Combinadic Ranking}
As with the Combinadic unranking, one can also reduce the time complexity 
of the Combinadic ranking (Alg.~\ref{alg:rankingcombinadic}) from $O(N^2)$ 
to $O(N)$ by computing ${c_i+1\choose i}$ and ${c_i+1\choose i+1}$ 
from ${c_i\choose i}$ in $O(1)$ time rather than from scratch in $O(i)$
time with the multiplicative formula (Eq.~\ref{eq:multiplicativeformula}). 
However, these $O(1)$-time properties require the values in the array $c$
to be consecutive, which can not be the case of OFDM-IM because these values
depend on the data the user transmits. One can avoid calculating all $k$ binomial 
coefficients from scratch by relying on the fact that the values 
$c_k>\cdots>c_2>c_1$ are restricted  to the integer range $[0, N-1]$.
Based on this, the linear-time Combinadic ranking (Alg.~\ref{alg:ranking2})
computes from scratch only one binomial coefficient (we refer 
to as $ccBinCoef$, line~\ref{ln:ranking2initializebc}) from which at most $N-1$ other
coefficients can be computed sequentially in $O(1)$ time each.  Since the
value of all other coefficients is computed from $ccBinCoef$, this variable
cannot be initialized with null binomial coefficients i.e., ${c_i\choose i}$
such that $c_i<i$. Thus, from lines~\ref{ln:ranking2specialcaseinit} to
\ref{ln:ranking2specialcaseend},  Alg.~\ref{alg:ranking2} looks for the largest
$i$ in the range $[0,\cdots, i,\cdots, N-1]$ such that $c_i\geq i$. 
These lines take $O(k)$ iterations. In line~\ref{ln:ranking2initializebc},
$ccBinCoef$ is initialized as ${c_i\choose i}$ in $O(i)$ time, yielding a
cumulative complexity of $O(k)+O(k)=O(k)$. From this, any consecutive
binomial coefficient (either ${c_i+1\choose i}$ or ${c_i+1\choose i+1}$)
can be computed in $O(1)$ time from $ccBinCoef={c_i\choose i}$ as in
the linear-time unranking algorithm.  Since the total number of remainder
binomial coefficients ranges from $i$ to $N-1$, the loop in line~\ref{ln:rank2for}
computes all of them in $O(N-i)=O(N)$ time. Therefore, the overall complexity is
$O(k)+O(k)+O(N)$ which becomes $O(N)$ under the optimal OFDM-IM setup (i.e., $k=N/2$).

\subsubsection{Scalable OFDM-IM Mapper Throughput}\label{subsubsec:ofdmimoptimalsc} 
We now proceed with the SC analysis of the optimal OFDM-IM mapper (Fig.~\ref{fig:optimalmapper})
considering an IxS implementation that meets our Theorem~\ref{th:imorder}. 
%We refer to this design as ``proposed mapper''. 
The analysis is as in subsection~\ref{subsubsec:ofdmimsc}, except for the fact that the 
IxS algorithm runs in $\Theta(N)$ time complexity.  Thus, the SC throughput is given by
%in Ineq.~\ref{eq:sc2b}. 
\begin{eqnarray}
\lim_{N\to\infty} \frac{N/2 + N-\log_2\sqrt N}{\kappa\cdot N} \label{eq:sc2b} %&=& \frac{3}{2\kappa} >0 
\end{eqnarray}

As $N$ grows, the time complexity is 
bounded by $\kappa N$ for some constant $\kappa>0$. Similarly, the SC 
throughput of the mapper results in a non-null constant $\kappa>0$, 
meeting the Def.~\ref{def:sca}. As explained in the  subsection~\ref{subsubsec:ofdmimsc},
$\kappa>0$ is constant that depends on the computational apparatus running the algorithm.
Under the linear-time IxS complexity, the throughput of the OFDM-IM mapper 
does not nullify for arbitrarily large $N$, %(Ineq.~\ref{eq:sc2b}).
\begin{eqnarray}
\lim_{N\to\infty} \frac{N/2 + N-\log_2\sqrt N}{\kappa\cdot N} &=& \frac{3}{2\kappa} >0 \label{eq:sc2bb} 
\end{eqnarray}

Note also that the throughput can increase with $N$ if one achieves a 
$o(N)$ mapper. However, as demonstrated in Corollary~\ref{col:lowerboundtradeoff},
this conflicts with the optimal SE setup, thereby preventing the SE maximization.

\section{Implementation and Evaluation}\label{sec:practical}
In this section, we present a practical case study to validate our 
theoretical findings. In subsection~\ref{sub:library}, we
introduce the open-source library we develop for the case study.
In subsection~\ref{subsec:methodology}, we describe the methodology
to assess and reproduce the empirical values of our experiments.
Finally, in subsection~\ref{subsec:results}, we present the results
of our practical case study that validate our theoretical findings.

% summarized in the Theorem~\ref{th:imorder}.

\begin{comment}
{\color{red}
We assess the spectro-computational throughput of the SE optimal  
OFDM-IM mapper (Fig.~\ref{fig:optimalmapper}), i.e., e.g. $g=1, k=N/2, M=2$,
varying the IxS algorithm and the number $N$ of subcarriers. 

}
\end{comment}
\subsection{Open-source OFDM-IM Mapper Library}\label{sub:library}
We wrote a C++ library that implements all OFDM-IM steps to map/demap
an $N$-subcarrier complex frequency-domain symbol. We implement the IxS block 
with C++ \emph{callbacks} to enable flexible addition of novel (un)ranking 
algorithms. In the released version,  we implement the original IxS 
algorithm~\cite{basar-ofdmim-transac-2013} and all the algorithms presented
in this work (Algs.~\ref{alg:unranking2} and~\ref{alg:ranking2}). We do not 
implement (un)ranking algorithms that can reach a complexity that is asymptotically 
faster than required by our Theorem~\ref{th:imorder} 
e.g.\cite{parque-2018, shimizu-unrakingsmalk-2014}. As previously 
explained~(Corollary~\ref{col:lowerboundtradeoff}),
performing (un)ranking faster than $\Theta(N)$ would require $k\not=N/2$,
thereby preventing the SE maximization (Corollary~\ref{col:lowerboundtradeoff}).
However, future works may implement IxS algorithms that improve
the original OFDM-IM using other criteria (e.g. BER~\cite{yoon-hammingmapper-2019}, 
\cite{wen-equiprobable-2016}.) than CC and SE. These and other IxS algorithms can 
also be included/evaluated in our library.
The entire source code of our library, as well as detailed instructions on
how to enhance it with novel IxS algorithms, are publicly available under the 
GPLv2 license in~\cite{queiroz-gitofdmim-2020}.
\begin{figure*}[t]
    \begin{subfigure}{.5\textwidth}
      \includegraphics[width=3in]{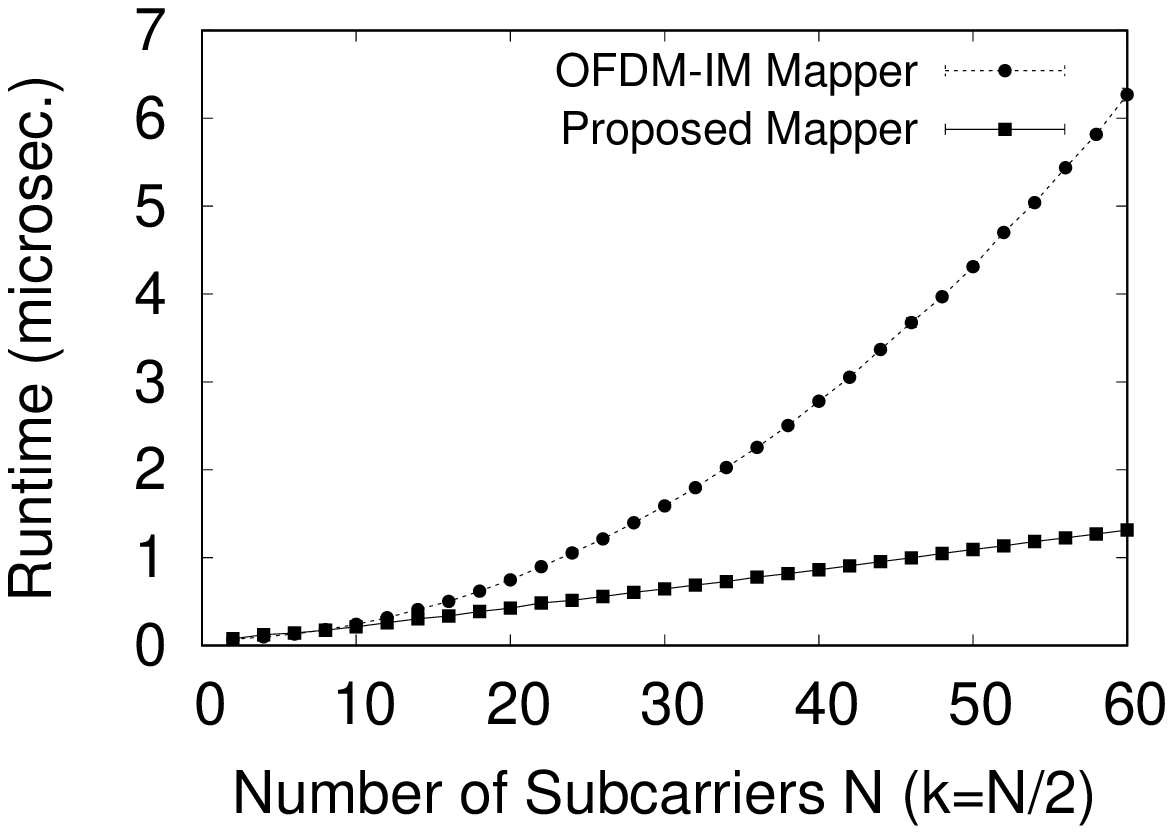}
      \caption{Runtime.\label{fig:runtimemapper}}
    \end{subfigure}
%\hfill
    \begin{subfigure}{.5\textwidth}
      \includegraphics[width=3in]{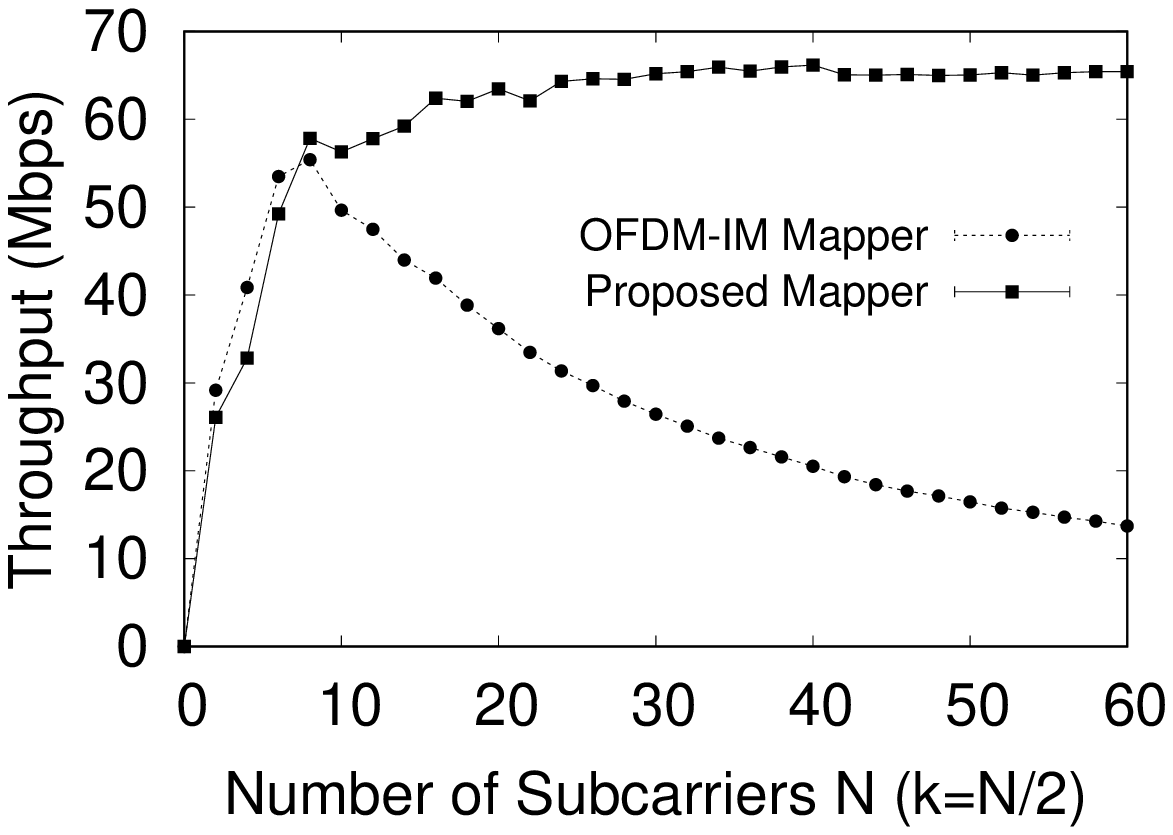}
      \caption{Throughput.\label{fig:throughputmmaper}}
    \end{subfigure}
      \caption{Mapper performance: Proposed vs. OFDM-IM. \label{fig:mapperperformance}}
\end{figure*}

\subsection{Performance Assessment Methodology}\label{subsec:methodology}
%In this section, we present a software-based case study for our theoretical findings.
We assess the runtime $T(N)$ (in secs.) and the throughput~$m(N)/T(N)$ 
(in megabits per seconds, Def.~\ref{def:sct}) for both the
original OFDM-IM mapper and our proposed mapper 
under the optimal SE configuration (i.e., $g=1$, $k=N/2$ and $M=2$).
For each mapper, we assess the performance indicators at both the 
transmitter (mapper) and the receiver (demapper) on a 3.5-GHz Intel 
i7-3770K processor.

We sampled the wall-clock runtime $T(N)$ of each mapper with the
standard C++ \texttt{timespace} library~\cite{timespec-2018} under the profile 
\texttt{CLOCK\_MONOTONIC}.
In each execution, we assigned our process with the largest real-time priority 
and employed the \texttt{isolcpus} Linux kernel directive to allocate one 
physical CPU core exclusively for each process. 
We generate the input for the mappers with the standard C++ 64-bit version of the
Mersenne Twistter (MT) 19937 pseudo-random number generator~\cite{matsumoto-mt-1998}. 
We set up three independent instances 
of MT19937\_64 with seeds $1973272912$, $1822174485$ and $1998078925$~\cite{prngs-2002}.
Every iteration, three sampled $T(N)$ are forwarded to the Akaroa-2 tool~\cite{akaroa2-2010} 
for statistical treatment. 
%Our C++ program keeps feeding Akaroa-2 in
%an infinite horizon simulation until the steady-state mean reaches the required precision.
Akaroa-2 determines the minimum number of samples required to reach the 
steady-state mean estimation of a given precision. 
In our experiments, this precision corresponds to a relative error below $5\%$ and a
confidence interval of $95\%$. Besides, in all experiments the highest observed 
variance was below $10^{-3}$ and the average number of samples in the transient 
state was about 300.

Table~\ref{tb:mapper} reports all assessed results for both the original
 OFDM-IM mapper and the proposed mapper at the transmitter (mapper).
The table~\ref{tb:demapper} reports the analogous results  assessed at the 
receiver (demapper). From left to right, the tables present the following
columns: the number $N$ of symbol's subcarriers, the number $m(N)$ of bits 
per symbol, the SE gain of the original OFDM-IM waveform against the classic 
OFDM mapper\footnote{The maximum SE gain is $m(N)/N$~\cite{fan-ofdmgim3-2015}.}, 
the assessed (de)mapper, the assessed runtime $T(N)$,  
the half-width of the confidence interval $\delta$ 
for $T(N)$,  the achieved (de)mapping throughput, and the number $x$ 
of samples needed to achieve the required precision. The source-code 
of all our experiments is publicly available under 
GPLv2 license in~\cite{queiroz-gitofdmim-2020}.

\begin{figure*}[t]
    \begin{subfigure}{.5\textwidth}
      \includegraphics[width=3in]{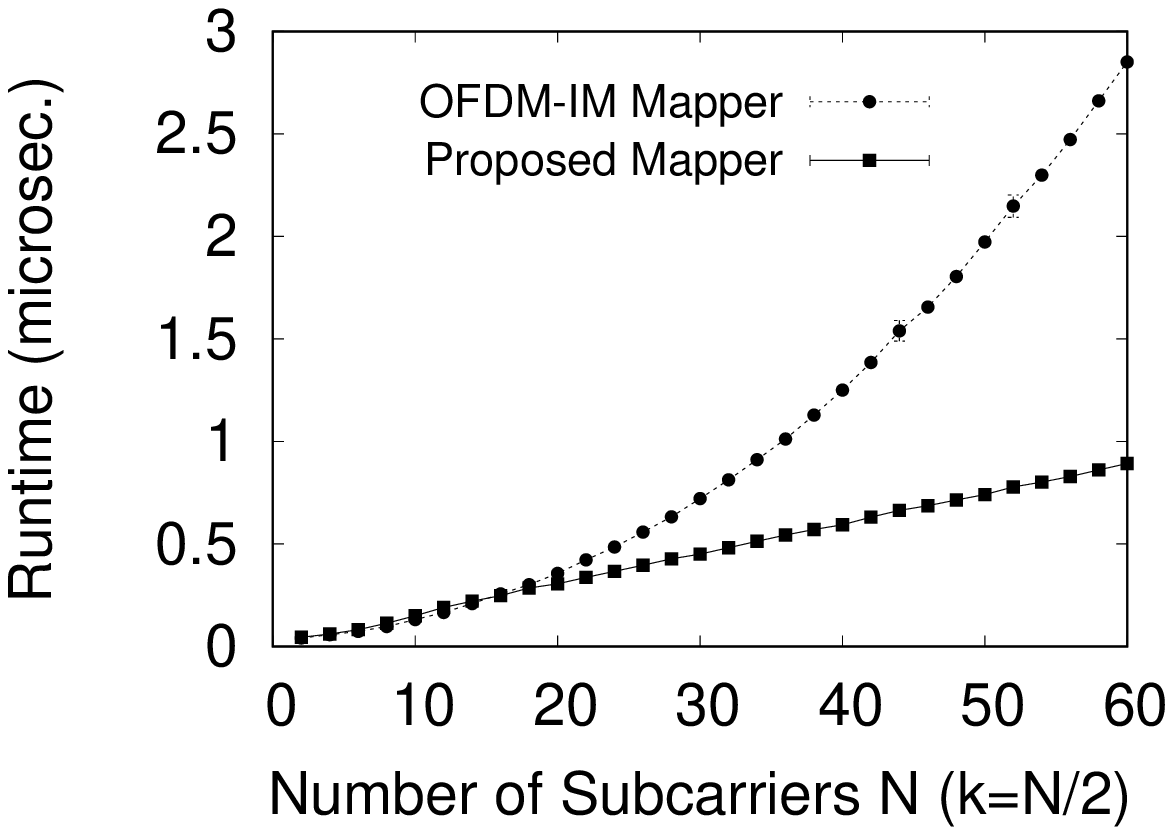}
      \caption{Runtime.\label{fig:runtimedemapper}}
    \end{subfigure}
%\hfill
    \begin{subfigure}{.5\textwidth}
      \includegraphics[width=3in]{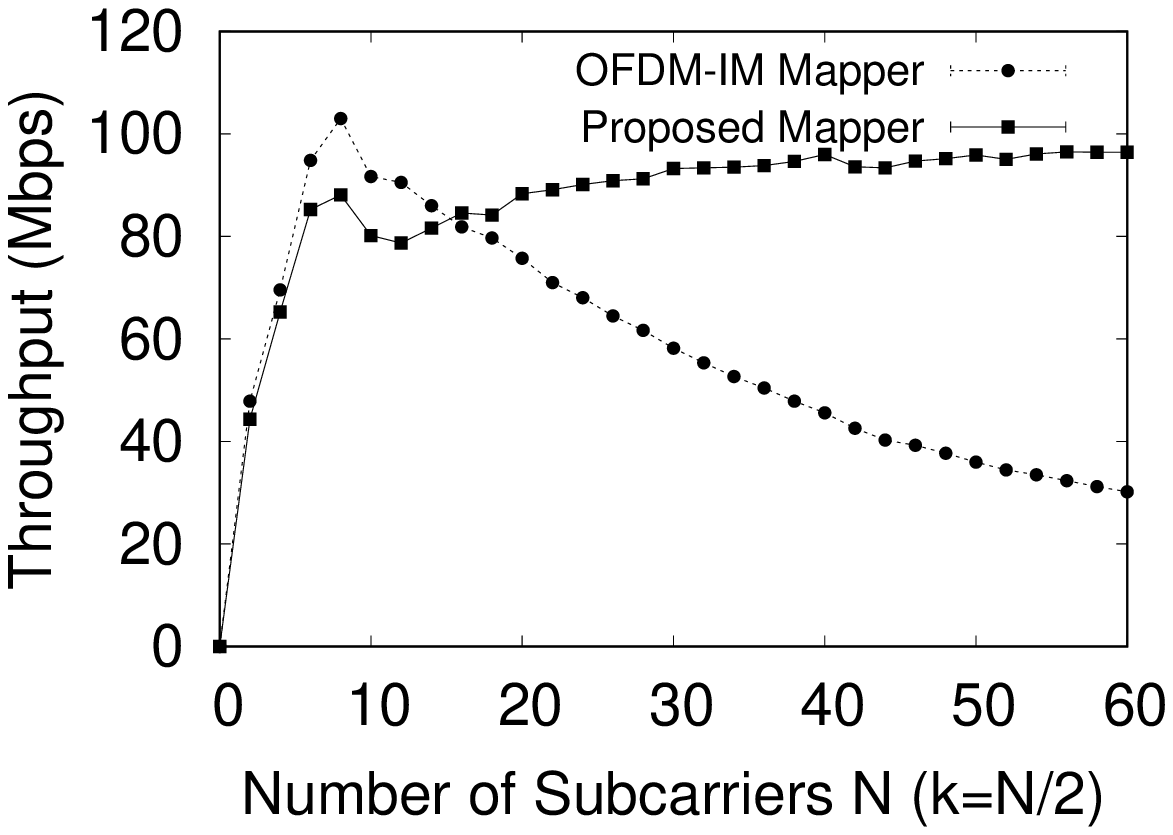}
      \caption{Throughput.\label{fig:throughputdemmaper}}
    \end{subfigure}
      \caption{Demapper performance: Proposed vs. OFDM-IM.\label{fig:demapperperformance}}
\end{figure*}

\subsection{Results}~\label{subsec:results}
%In Fig.~\ref{fig:mapperperformance} we report the performance of the
%compared mappers.
In Fig.~\ref{fig:runtimemapper} and Fig.~\ref{fig:throughputmmaper}, we respectively 
plot the runtime and the throughput performances of the compared mappers for 
$N=2,4,\dots,62$. Although only particular values of $N$ verify in industry
standards (e.g. $N=48$~\cite{ieee80211-12}, $N=52$~\cite{ieee80211ac-2013}), 
we range it from small to large values to illustrate the asymptotic shape
predicted by our throughput analysis.
Detailed information about these plots are reported on the Table~\ref{tb:mapper}.
As predicted by our theoretical analysis (Subsections~\ref{subsubsec:ofdmimsc} 
and~\ref{subsec:optimal}), in the ideal setup, the runtime order of growth of 
the original OFDM-IM mapper is asymptotically larger than our proposed 
mapper (Fig.~\ref{fig:runtimemapper}). From the theoretical analysis, we know 
these complexities are $O(N^2)$ and $O(N)$, respectively. Naturally, the
runtime curves of both mappers increase monotonically towards infinite
as the number $N$ of subcarriers grows. However,
because the runtime order of growth of the original OFDM-IM mapper is larger
than the number $m(N)=N/2+\log_2{N\choose N/2}=O(N)$ of bits per symbol,
the throughput $m(N)/T(N)$ of this mapper nullifies as $N$ grows 
(Fig.~\ref{fig:throughputmmaper}). This validates the theoretical analysis
we show in subsection~\ref{subsubsec:ofdmimsc}.

By contrast, when our proposed mapper takes place, both the resulting
computational complexity $T(N)$ and the total of bits $m(N)$ per symbol 
increases in the same order of growth. Thus, the throughput
$m(N)/T(N)$ tends to a non-null constant. In particular, according to our 
theoretical analysis in subsection~\ref{subsubsec:ofdmimoptimalsc},
this is  $m(N)/T(N)=3/(2\kappa)$. Recall that the constant $\kappa>0$ captures the
 wall-clock runtime taken by the asymptotic dominant instruction of the algorithm 
on a real machine. However, in our practical case study, the assessed runtime 
$T(N)$ encompasses all computational instructions performed by each (de)mapper. 
Thus, $\kappa$ represents an average of the runtime taken by each kind of
instruction on the machine of our testbed i.e., the Intel i7-3770K 
processor. From the assessed throughput $m(N)/T(N)$, the average value of 
$\kappa$ can be computed based on Eq.~(\ref{eq:sc2bb}), which is 
$\kappa=3/2\cdot 1/(m(N)/T(N))$. In our testbed, the average runtime 
per computational instruction was $0.02$~$\mu s$.

In Fig.~\ref{fig:runtimedemapper} and Fig.~\ref{fig:throughputdemmaper}, we respectively 
plot the runtime and the throughput performances of the compared demappers for different 
values of $N$. Detailed informations of these plots are reported on the 
Table~\ref{tb:demapper}. As in the mapper analysis, the throughput of the 
original OFDM-IM demapper tends to zero as $N$ grows whereas the throughput of
our proposed demapper tends to a non-null constant under the same conditions.
If compared against its corresponding mapper, we verify that our proposed demapper 
presents larger throughput. This means that, although both our mapper and demapper have
the same $O(N)$ asymptotic complexity, the demapper implementation is less complex 
concerning the constant $\kappa$. Indeed, we verify an average $\kappa=0.015$~$\mu s$
for the demapper in contrast with the $0.02$~$\mu s$ for the mapper.

\section{Conclusion and Future Directions}\label{sec:conclusion}
In this work, we studied the trade-off between spectral efficiency (SE) 
and computational complexity (CC) $T(N)$ of an $N$-subcarrier OFDM with Index 
Modulation (OFDM-IM) mapper. We identified that the CC lower bound to map any 
of all $2^{\lfloor\log_2{N \choose N/2}\rfloor}$ OFDM-IM waveforms 
is $\Omega(N)$. With this, we formally proved that enabling all OFDM-IM 
waveforms is not computationally intractable, as previously 
conjectured~\cite{lu-compressiveim-2018,basar-thesurvey-2017}. Besides,
we showed that any algorithm running faster than this lower bound
prevents the OFDM-IM SE maximization.
We also presented the spectro-computational efficiency (SCE) metric 
both to analyze the mapper's throughput and identify an upper bound
for the mapper's complexity $T(N)$ under the maximal SE. In this context,
we proved that the worst tolerable CC for the mapper is $O(N)$ otherwise,
the mapper's throughput nullifies as the system is assigned more and
more subcarriers. We showed that this is the case of the original OFDM-IM 
mapper~\cite{basar-ofdmim-transac-2013}, in which the $O(N^2)$
CC surpasses the $O(N\log_2N)$ CC of the IDFT/DFT algorithm. Then, we
presented an OFDM-IM mapper that enables the largest SE
under the minimum possible CC. 

We demonstrate our theoretical findings by implementing 
an open-source library that supports all DSP steps to map/demap an $N$-subcarrier 
complex frequency-domain OFDM-IM symbol. Our implementation supports different 
index selector algorithms and is the first to enable the SE maximization while 
preserving the same time and space asymptotic complexities of the classic OFDM mapper.
With our library, \textit{we showed that the OFDM-IM mapper does not need 
compromise approaches that prevail in the OFDM-IM literature such as 
subblock partitioning (SP)~\cite{mao-survey-ofdmim-2018, 
mao-dm_ofdm_im-ieeeaccess-2017,basar-thesurvey-2017, fan-ofdmgim3-2015, 
basar-ofdmim-transac-2013}, adoption of few active subcarriers~\cite{salah-2019}
or extra space complexity~\cite{queiroz-wcl-19}}.

Future works may consider extra performance indicators in the analysis 
(in addition to CC and SE) such as bit-error rate~\cite{yoon-hammingmapper-2019, wen-equiprobable-2016}. 
Moreover, our mapper can be directly applied to other IM systems that rely on the
same index selector of the original OFDM-IM mapper such as spatial modulation 
systems~\cite{jsac-spatialim-2019} and dual mode OFDM-IM~\cite{mao-dm_ofdm_im-ieeeaccess-2017}.
Besides, our methodology can guide the activation of all waveforms in other variants
of OFDM-IM such as multiple mode OFDM-IM~\cite{wen-mmode_ofdm_im-tran_comm-17}.
\begin{table}
\caption{Mapper performance: Proposed (``Prop.'') vs. original OFDM-IM (``Orig.'') \label{tb:mapper}
%Mean runtime of mapping algorithms~\ref{alg:lixs} (LIxS) and~\ref{alg:ixs} (IxS)
%to map $m$ bits by selecting $N/2$ out of $N$ subcarriers. `IM gain' is the
%spectral efficiency gain of OFDM-IM over OFDM. For mean runtimes,
%$\delta$ is the half-width of the confidence interval,  $x$ is the number of samples 
%needed to achieve a relative error $<5\%$ and a confidence interval of $95\%$ and 
%$x*$ is the number of discarded samples before the system get at the steady-state.}
% All observed variance was below $10^{-3}$.
}
\centering
{\scriptsize
%\begin{tabular}{|>{\centering}m{1em}|>{\centering}m{1.7em}|c|c|>{\centering}m{3.1em}|c|c|c|}
\begin{tabular}{|c|c|c|c|c|c|c|c|}
\hline
$N$ 
& 
  \begin{tabular}[c] {@{}c@{}}
    $m(N)$  \\ bits
  \end{tabular} 
&
  \begin{tabular}[c] {@{}c@{}}
      IM  \\ Gain
  \end{tabular} 
&
  \begin{tabular}[c] {@{}c@{}}
      IM \\ Mapper
  \end{tabular} 
  & 
  \begin{tabular}[c]{@{}c@{}}
      \textbf{Runtime} \\$({\mu s})$
  \end{tabular} 
  & 
  \begin{tabular}[c]{@{}c@{}}
       $\pm\delta$ \\ $(\mu s)$
  \end{tabular} 
  &
  \begin{tabular}[c]{@{}c@{}}
       \textbf{Through}\\ \textbf{put (Mbps)}
  \end{tabular} 
  & 
  $x$ 
   \\ \hline
% PASTE LATEX-TABLE.TXT HERE
\multirow{2}{*}{2} & \multirow{2}{*}{2} & \multirow{2}{*}{1.00} &Prop. &\textbf{0.08} & 0.004 & \textbf{26.08}  & 3792 \\ \cline{4-8}
& & & Orig. & \textbf{0.07} & 0.002 & \textbf{29.15} & 2358 \\ \hline\hline
\multirow{2}{*}{4} & \multirow{2}{*}{4} & \multirow{2}{*}{1.00} &Prop. &\textbf{0.12} & 0.001 & \textbf{32.81}  & 1854 \\ \cline{4-8}
& & & Orig. & \textbf{0.10} & 0.002 & \textbf{40.86} & 1710 \\ \hline\hline
\multirow{2}{*}{6} & \multirow{2}{*}{7} & \multirow{2}{*}{1.17} &Prop. &\textbf{0.14} & 0.003 & \textbf{49.23}  & 1704 \\ \cline{4-8}
& & & Orig. & \textbf{0.13} & 0.003 & \textbf{53.48} & 1656 \\ \hline\hline
\multirow{2}{*}{8} & \multirow{2}{*}{10} & \multirow{2}{*}{1.25} &Prop. &\textbf{0.17} & 0.001 & \textbf{57.84}  & 1686 \\ \cline{4-8}
& & & Orig. & \textbf{0.18} & 0.001 & \textbf{55.40} & 1602 \\ \hline\hline
\multirow{2}{*}{10} & \multirow{2}{*}{12} & \multirow{2}{*}{1.20} &Prop. &\textbf{0.21} & 0.001 & \textbf{56.29}  & 1536 \\ \cline{4-8}
& & & Orig. & \textbf{0.24} & 0.002 & \textbf{49.65} & 2208 \\ \hline\hline
\multirow{2}{*}{12} & \multirow{2}{*}{15} & \multirow{2}{*}{1.25} &Prop. &\textbf{0.26} & 0.002 & \textbf{57.78}  & 1614 \\ \cline{4-8}
& & & Orig. & \textbf{0.32} & 0.003 & \textbf{47.47} & 1872 \\ \hline\hline
\multirow{2}{*}{14} & \multirow{2}{*}{18} & \multirow{2}{*}{1.28} &Prop. &\textbf{0.30} & 0.002 & \textbf{59.21}  & 1524 \\ \cline{4-8}
& & & Orig. & \textbf{0.41} & 0.003 & \textbf{43.99} & 1542 \\ \hline\hline
\multirow{2}{*}{16} & \multirow{2}{*}{21} & \multirow{2}{*}{1.31} &Prop. &\textbf{0.34} & 0.002 & \textbf{62.39}  & 1728 \\ \cline{4-8}
& & & Orig. & \textbf{0.50} & 0.004 & \textbf{41.92} & 1476 \\ \hline\hline
\multirow{2}{*}{18} & \multirow{2}{*}{24} & \multirow{2}{*}{1.33} &Prop. &\textbf{0.39} & 0.002 & \textbf{62.03}  & 1596 \\ \cline{4-8}
& & & Orig. & \textbf{0.62} & 0.005 & \textbf{38.85} & 1494 \\ \hline\hline
\multirow{2}{*}{20} & \multirow{2}{*}{27} & \multirow{2}{*}{1.35} &Prop. &\textbf{0.43} & 0.002 & \textbf{63.45}  & 1524 \\ \cline{4-8}
& & & Orig. & \textbf{0.75} & 0.007 & \textbf{36.18} & 1554 \\ \hline
\multirow{2}{*}{22} & \multirow{2}{*}{30} & \multirow{2}{*}{1.36} &Prop. &\textbf{0.48} & 0.002 & \textbf{62.10}  & 1884 \\ \cline{4-8}
& & & Orig. & \textbf{0.90} & 0.008 & \textbf{33.46} & 1518 \\ \hline\hline
\multirow{2}{*}{24} & \multirow{2}{*}{33} & \multirow{2}{*}{1.38} &Prop. &\textbf{0.51} & 0.002 & \textbf{64.31}  & 1554 \\ \cline{4-8}
& & & Orig. & \textbf{1.05} & 0.043 & \textbf{31.35} & 1512 \\ \hline\hline
\multirow{2}{*}{26} & \multirow{2}{*}{36} & \multirow{2}{*}{1.38} &Prop. &\textbf{0.56} & 0.001 & \textbf{64.61}  & 1560 \\ \cline{4-8}
& & & Orig. & \textbf{1.21} & 0.007 & \textbf{29.70} & 1470 \\ \hline\hline
\multirow{2}{*}{28} & \multirow{2}{*}{39} & \multirow{2}{*}{1.39} &Prop. &\textbf{0.60} & 0.002 & \textbf{64.55}  & 1536 \\ \cline{4-8}
& & & Orig. & \textbf{1.40} & 0.012 & \textbf{27.92} & 1512 \\ \hline\hline
\multirow{2}{*}{30} & \multirow{2}{*}{42} & \multirow{2}{*}{1.40} &Prop. &\textbf{0.64} & 0.003 & \textbf{65.19}  & 1518 \\ \cline{4-8}
& & & Orig. & \textbf{1.59} & 0.016 & \textbf{26.43} & 1476 \\ \hline\hline
\multirow{2}{*}{32} & \multirow{2}{*}{45} & \multirow{2}{*}{1.41} &Prop. &\textbf{0.69} & 0.010 & \textbf{65.43}  & 1524 \\ \cline{4-8}
& & & Orig. & \textbf{1.79} & 0.012 & \textbf{25.07} & 1548 \\ \hline\hline
\multirow{2}{*}{34} & \multirow{2}{*}{48} & \multirow{2}{*}{1.41} &Prop. &\textbf{0.73} & 0.003 & \textbf{65.93}  & 1560 \\ \cline{4-8}
& & & Orig. & \textbf{2.03} & 0.018 & \textbf{23.70} & 1518 \\ \hline\hline

\multirow{2}{*}{36} & \multirow{2}{*}{51} & \multirow{2}{*}{1.42} &Prop. &\textbf{0.78} & 0.008 & \textbf{65.47}  & 1500 \\ \cline{4-8}
& & & Orig. & \textbf{2.25} & 0.015 & \textbf{22.63} & 1482 \\ \hline\hline
\multirow{2}{*}{38} & \multirow{2}{*}{54} & \multirow{2}{*}{1.42} &Prop. &\textbf{0.82} & 0.002 & \textbf{65.96}  & 1608 \\ \cline{4-8}
& & & Orig. & \textbf{2.50} & 0.017 & \textbf{21.57} & 1776 \\ \hline\hline
\multirow{2}{*}{40} & \multirow{2}{*}{57} & \multirow{2}{*}{1.42} &Prop. &\textbf{0.86} & 0.002 & \textbf{66.16}  & 1524 \\ \cline{4-8}
& & & Orig. & \textbf{2.78} & 0.027 & \textbf{20.51} & 1530 \\ \hline\hline
\multirow{2}{*}{42} & \multirow{2}{*}{59} & \multirow{2}{*}{1.40} &Prop. &\textbf{0.91} & 0.003 & \textbf{65.06}  & 1620 \\ \cline{4-8}
& & & Orig. & \textbf{3.05} & 0.019 & \textbf{19.33} & 1458 \\ \hline\hline
\multirow{2}{*}{44} & \multirow{2}{*}{62} & \multirow{2}{*}{1.41} &Prop. &\textbf{0.95} & 0.003 & \textbf{65.02}  & 1686 \\ \cline{4-8}
& & & Orig. & \textbf{3.37} & 0.027 & \textbf{18.41} & 1518 \\ \hline\hline
\multirow{2}{*}{46} & \multirow{2}{*}{65} & \multirow{2}{*}{1.41} &Prop. &\textbf{1.00} & 0.002 & \textbf{65.10}  & 2118 \\ \cline{4-8}
& & & Orig. & \textbf{3.68} & 0.055 & \textbf{17.68} & 1548 \\ \hline\hline
\multirow{2}{*}{48} & \multirow{2}{*}{68} & \multirow{2}{*}{1.42} &Prop. &\textbf{1.05} & 0.002 & \textbf{64.98}  & 1536 \\ \cline{4-8}
& & & Orig. & \textbf{3.97} & 0.022 & \textbf{17.13} & 1476 \\ \hline\hline
\multirow{2}{*}{50} & \multirow{2}{*}{71} & \multirow{2}{*}{1.42} &Prop. &\textbf{1.09} & 0.010 & \textbf{65.04}  & 1530 \\ \cline{4-8}
& & & Orig. & \textbf{4.31} & 0.035 & \textbf{16.47} & 1494 \\ \hline\hline
\multirow{2}{*}{52} & \multirow{2}{*}{74} & \multirow{2}{*}{1.42} &Prop. &\textbf{1.13} & 0.002 & \textbf{65.31}  & 1578 \\ \cline{4-8}
& & & Orig. & \textbf{4.70} & 0.022 & \textbf{15.75} & 1494 \\ \hline\hline
\multirow{2}{*}{54} & \multirow{2}{*}{77} & \multirow{2}{*}{1.42} &Prop. &\textbf{1.18} & 0.002 & \textbf{65.03}  & 1470 \\ \cline{4-8}
& & & Orig. & \textbf{5.04} & 0.025 & \textbf{15.28} & 1500 \\ \hline\hline
\multirow{2}{*}{56} & \multirow{2}{*}{80} & \multirow{2}{*}{1.43} &Prop. &\textbf{1.23} & 0.002 & \textbf{65.31}  & 1440 \\ \cline{4-8}
& & & Orig. & \textbf{5.44} & 0.026 & \textbf{14.71} & 1536 \\ \hline\hline
\multirow{2}{*}{58} & \multirow{2}{*}{83} & \multirow{2}{*}{1.43} &Prop. &\textbf{1.27} & 0.004 & \textbf{65.42}  & 2064 \\ \cline{4-8}
& & & Orig. & \textbf{5.82} & 0.035 & \textbf{14.27} & 1512 \\ \hline\hline
\multirow{2}{*}{60} & \multirow{2}{*}{86} & \multirow{2}{*}{1.43} &Prop. &\textbf{1.31} & 0.003 & \textbf{65.42}  & 1614 \\ \cline{4-8}
& & & Orig. & \textbf{6.27} & 0.073 & \textbf{13.72} & 1476 \\ \hline\hline
\multirow{2}{*}{62} & \multirow{2}{*}{89} & \multirow{2}{*}{1.43} &Prop. &\textbf{1.36} & 0.003 & \textbf{65.45}  & 1548 \\ \cline{4-8}
& & & Orig. & \textbf{6.70} & 0.027 & \textbf{13.28} & 1500 \\ \hline
%\multirow{2}{*}{66} & \multirow{2}{*}{95} & \multirow{2}{*}{1.44} &Prop. &\textbf{1.41} & 0.008 & \textbf{67.54}  & 1554 \\ \cline{4-8}
%& & & Orig. & \textbf{8.95} & 0.036 & \textbf{10.62} & 1518 \\ \hline
\arrayrulecolor{white} $\vdots$ & $\vdots$  & $\vdots$  & $\vdots$  & $\vdots$ & $\vdots$  & $\vdots$  &$\vdots$ \\ \hline \hline
\arrayrulecolor{black}\hline
\multirow{2}{*}{$\infty$} & \multirow{2}{*}{\scriptsize${\Theta(N)}$} & \multirow{2}{*}{1.5} &Prop. &{\scriptsize$\Theta(N)$} & 0 & $3/(2\kappa)$ & $\infty$ \\ \cline{4-8}
& & & Orig. & {\scriptsize$\Theta(N^2)$} & 0 & $0$ & $\infty$ \\ \hline
\end{tabular}
}
\end{table}
\begin{table}
\caption{Demapper performance: Proposed (``Prop.'') vs. original OFDM-IM (``Orig.'').\label{tb:demapper}
%Mean runtime of mapping algorithms~\ref{alg:lixs} (LIxS) and~\ref{alg:ixs} (IxS)
%to map $m$ bits by selecting $N/2$ out of $N$ subcarriers. `IM gain' is the
%spectral efficiency gain of OFDM-IM over OFDM. For mean runtimes,
%$\delta$ is the half-width of the confidence interval,  $x$ is the number of samples 
%needed to achieve a relative error $<5\%$ and a confidence interval of $95\%$ and 
%$x*$ is the number of discarded samples before the system get at the steady-state.}
% All observed variance was below $10^{-3}$.
}
\centering
\label{tb:talpha}
{\scriptsize
%\begin{tabular}{|>{\centering}m{1em}|>{\centering}m{1.7em}|c|c|>{\centering}m{3.1em}|c|c|c|}
\begin{tabular}{|c|c|c|c|c|c|c|c|}
\hline
$N$ 
& 
  \begin{tabular}[c] {@{}c@{}}
    $m(N)$  \\ bits
  \end{tabular} 
&
  \begin{tabular}[c] {@{}c@{}}
      IM  \\ Gain
  \end{tabular} 
&
  \begin{tabular}[c] {@{}c@{}}
      IM \\ Mapper
  \end{tabular} 
  & 
  \begin{tabular}[c]{@{}c@{}}
      \textbf{Runtime} \\$({\mu s})$
  \end{tabular} 
  & 
  \begin{tabular}[c]{@{}c@{}}
       $\pm\delta$ \\ $(\mu s)$
  \end{tabular} 
  &
  \begin{tabular}[c]{@{}c@{}}
       \textbf{Through}\\ \textbf{put (Mbps)}
  \end{tabular} 
  & 
  $x$ 
   \\ \hline
% PASTE LATEX-TABLE.TXT HERE
\multirow{2}{*}{2} & \multirow{2}{*}{2} & \multirow{2}{*}{1.00} &Prop. &\textbf{0.05} & 0.001 & \textbf{44.35}  & 1644 \\ \cline{4-8}
& & & Orig. & \textbf{0.04} & 0.001 & \textbf{47.85} & 1680 \\ \hline\hline
\multirow{2}{*}{4} & \multirow{2}{*}{4} & \multirow{2}{*}{1.00} &Prop. &\textbf{0.06} & 0.001 & \textbf{65.25}  & 1674 \\ \cline{4-8}
& & & Orig. & \textbf{0.06} & 0.000 & \textbf{69.57} & 1620 \\ \hline\hline
\multirow{2}{*}{6} & \multirow{2}{*}{7} & \multirow{2}{*}{1.17} &Prop. &\textbf{0.08} & 0.001 & \textbf{85.26}  & 1758 \\ \cline{4-8}
& & & Orig. & \textbf{0.07} & 0.001 & \textbf{94.85} & 1746 \\ \hline\hline
\multirow{2}{*}{8} & \multirow{2}{*}{10} & \multirow{2}{*}{1.25} &Prop. &\textbf{0.11} & 0.001 & \textbf{88.11}  & 2208 \\ \cline{4-8}
& & & Orig. & \textbf{0.10} & 0.001 & \textbf{102.99} & 1626 \\ \hline\hline
\multirow{2}{*}{10} & \multirow{2}{*}{12} & \multirow{2}{*}{1.20} &Prop. &\textbf{0.15} & 0.002 & \textbf{80.16}  & 1518 \\ \cline{4-8}
& & & Orig. & \textbf{0.13} & 0.001 & \textbf{91.67} & 1704 \\ \hline\hline
\multirow{2}{*}{12} & \multirow{2}{*}{15} & \multirow{2}{*}{1.25} &Prop. &\textbf{0.19} & 0.002 & \textbf{78.70}  & 1524 \\ \cline{4-8}
& & & Orig. & \textbf{0.17} & 0.002 & \textbf{90.53} & 1536 \\ \hline\hline
\multirow{2}{*}{14} & \multirow{2}{*}{18} & \multirow{2}{*}{1.28} &Prop. &\textbf{0.22} & 0.001 & \textbf{81.63}  & 1536 \\ \cline{4-8}
& & & Orig. & \textbf{0.21} & 0.001 & \textbf{86.00} & 1614 \\ \hline\hline
\multirow{2}{*}{16} & \multirow{2}{*}{21} & \multirow{2}{*}{1.31} &Prop. &\textbf{0.25} & 0.002 & \textbf{84.58}  & 1512 \\ \cline{4-8}
& & & Orig. & \textbf{0.26} & 0.001 & \textbf{81.87} & 1758 \\ \hline\hline
\multirow{2}{*}{18} & \multirow{2}{*}{24} & \multirow{2}{*}{1.33} &Prop. &\textbf{0.29} & 0.003 & \textbf{84.18}  & 1536 \\ \cline{4-8}
& & & Orig. & \textbf{0.30} & 0.002 & \textbf{79.68} & 1524 \\ \hline\hline
\multirow{2}{*}{20} & \multirow{2}{*}{27} & \multirow{2}{*}{1.35} &Prop. &\textbf{0.31} & 0.002 & \textbf{88.32}  & 1542 \\ \cline{4-8}
& & & Orig. & \textbf{0.36} & 0.001 & \textbf{75.74} & 1614 \\ \hline\hline
\multirow{2}{*}{22} & \multirow{2}{*}{30} & \multirow{2}{*}{1.36} &Prop. &\textbf{0.34} & 0.004 & \textbf{89.13}  & 1596 \\ \cline{4-8}
& & & Orig. & \textbf{0.42} & 0.002 & \textbf{70.99} & 1542 \\ \hline\hline
\multirow{2}{*}{24} & \multirow{2}{*}{33} & \multirow{2}{*}{1.38} &Prop. &\textbf{0.37} & 0.003 & \textbf{90.11}  & 1488 \\ \cline{4-8}
& & & Orig. & \textbf{0.49} & 0.007 & \textbf{68.03} & 1704 \\ \hline\hline
\multirow{2}{*}{26} & \multirow{2}{*}{36} & \multirow{2}{*}{1.38} &Prop. &\textbf{0.40} & 0.003 & \textbf{90.86}  & 1500 \\ \cline{4-8}
& & & Orig. & \textbf{0.56} & 0.001 & \textbf{64.49} & 1530 \\ \hline\hline
\multirow{2}{*}{28} & \multirow{2}{*}{39} & \multirow{2}{*}{1.39} &Prop. &\textbf{0.43} & 0.008 & \textbf{91.23}  & 1482 \\ \cline{4-8}
& & & Orig. & \textbf{0.63} & 0.003 & \textbf{61.69} & 1626 \\ \hline\hline
\multirow{2}{*}{30} & \multirow{2}{*}{42} & \multirow{2}{*}{1.40} &Prop. &\textbf{0.45} & 0.002 & \textbf{93.23}  & 1566 \\ \cline{4-8}
& & & Orig. & \textbf{0.72} & 0.001 & \textbf{58.20} & 1548 \\ \hline\hline
\multirow{2}{*}{32} & \multirow{2}{*}{45} & \multirow{2}{*}{1.41} &Prop. &\textbf{0.48} & 0.007 & \textbf{93.38}  & 1464 \\ \cline{4-8}
& & & Orig. & \textbf{0.81} & 0.002 & \textbf{55.35} & 1476 \\ \hline\hline
\multirow{2}{*}{34} & \multirow{2}{*}{48} & \multirow{2}{*}{1.41} &Prop. &\textbf{0.51} & 0.003 & \textbf{93.51}  & 1602 \\ \cline{4-8}
& & & Orig. & \textbf{0.91} & 0.001 & \textbf{52.67} & 1560 \\ \hline\hline
\multirow{2}{*}{36} & \multirow{2}{*}{51} & \multirow{2}{*}{1.42} &Prop. &\textbf{0.54} & 0.011 & \textbf{93.82}  & 1548 \\ \cline{4-8}
& & & Orig. & \textbf{1.01} & 0.002 & \textbf{50.42} & 1878 \\ \hline\hline
\multirow{2}{*}{38} & \multirow{2}{*}{54} & \multirow{2}{*}{1.42} &Prop. &\textbf{0.57} & 0.002 & \textbf{94.62}  & 1500 \\ \cline{4-8}
& & & Orig. & \textbf{1.13} & 0.002 & \textbf{47.85} & 1512 \\ \hline\hline
\multirow{2}{*}{40} & \multirow{2}{*}{57} & \multirow{2}{*}{1.42} &Prop. &\textbf{0.59} & 0.003 & \textbf{95.99}  & 1464 \\ \cline{4-8}
& & & Orig. & \textbf{1.25} & 0.003 & \textbf{45.58} & 1548 \\ \hline\hline
\multirow{2}{*}{42} & \multirow{2}{*}{59} & \multirow{2}{*}{1.40} &Prop. &\textbf{0.63} & 0.003 & \textbf{93.58}  & 1548 \\ \cline{4-8}
& & & Orig. & \textbf{1.39} & 0.014 & \textbf{42.58} & 1722 \\ \hline\hline
\multirow{2}{*}{44} & \multirow{2}{*}{62} & \multirow{2}{*}{1.41} &Prop. &\textbf{0.66} & 0.002 & \textbf{93.37}  & 1464 \\ \cline{4-8}
& & & Orig. & \textbf{1.54} & 0.051 & \textbf{40.27} & 2259 \\ \hline\hline
\multirow{2}{*}{46} & \multirow{2}{*}{65} & \multirow{2}{*}{1.41} &Prop. &\textbf{0.69} & 0.001 & \textbf{94.68}  & 1512 \\ \cline{4-8}
& & & Orig. & \textbf{1.66} & 0.002 & \textbf{39.27} & 1656 \\ \hline\hline
\multirow{2}{*}{48} & \multirow{2}{*}{68} & \multirow{2}{*}{1.42} &Prop. &\textbf{0.71} & 0.006 & \textbf{95.18}  & 1554 \\ \cline{4-8}
& & & Orig. & \textbf{1.80} & 0.006 & \textbf{37.68} & 1548 \\ \hline\hline
\multirow{2}{*}{50} & \multirow{2}{*}{71} & \multirow{2}{*}{1.42} &Prop. &\textbf{0.74} & 0.002 & \textbf{95.88}  & 1458 \\ \cline{4-8}
& & & Orig. & \textbf{1.97} & 0.013 & \textbf{35.98} & 1488 \\ \hline\hline
\multirow{2}{*}{52} & \multirow{2}{*}{74} & \multirow{2}{*}{1.42} &Prop. &\textbf{0.78} & 0.009 & \textbf{95.05}  & 1530 \\ \cline{4-8}
& & & Orig. & \textbf{2.15} & 0.055 & \textbf{34.45} & 1506 \\ \hline\hline
\multirow{2}{*}{54} & \multirow{2}{*}{77} & \multirow{2}{*}{1.42} &Prop. &\textbf{0.80} & 0.002 & \textbf{96.08}  & 1542 \\ \cline{4-8}
& & & Orig. & \textbf{2.30} & 0.002 & \textbf{33.49} & 1548 \\ \hline\hline
\multirow{2}{*}{56} & \multirow{2}{*}{80} & \multirow{2}{*}{1.43} &Prop. &\textbf{0.83} & 0.002 & \textbf{96.50}  & 1572 \\ \cline{4-8}
& & & Orig. & \textbf{2.47} & 0.003 & \textbf{32.35} & 1512 \\ \hline\hline
\multirow{2}{*}{58} & \multirow{2}{*}{83} & \multirow{2}{*}{1.43} &Prop. &\textbf{0.86} & 0.003 & \textbf{96.42}  & 1566 \\ \cline{4-8}
& & & Orig. & \textbf{2.66} & 0.002 & \textbf{31.19} & 1524 \\ \hline\hline
\multirow{2}{*}{60} & \multirow{2}{*}{86} & \multirow{2}{*}{1.43} &Prop. &\textbf{0.89} & 0.003 & \textbf{96.41}  & 1482 \\ \cline{4-8}
& & & Orig. & \textbf{2.85} & 0.002 & \textbf{30.16} & 1506 \\ \hline
\multirow{2}{*}{62} & \multirow{2}{*}{89} & \multirow{2}{*}{1.43} &Prop. &\textbf{0.92} & 0.003 & \textbf{96.33}  & 1614 \\ \cline{4-8}
& & & Orig. & \textbf{3.09} & 0.064 & \textbf{28.76} & 1488 \\ \hline%\multirow{2}{*}{66} & \multirow{2}{*}{95} & \multirow{2}{*}{1.44} &Prop. &\textbf{0.70} & 0.002 & \textbf{135.87}  & 1542 \\ \cline{4-8}
%& & & Orig. & \textbf{3.91} & 0.048 & \textbf{24.31} & 1542 \\ \hline\hline
\arrayrulecolor{white} $\vdots$ & $\vdots$  & $\vdots$  & $\vdots$  & $\vdots$ & $\vdots$  & $\vdots$  &$\vdots$ \\ \hline \hline
\arrayrulecolor{black}\hline
\multirow{2}{*}{$\infty$} & \multirow{2}{*}{\scriptsize${\Theta(N)}$} & \multirow{2}{*}{1.5} &Prop. &{\scriptsize$\Theta(N)$} & 0 & $3/(2\kappa)$ & $\infty$ \\ \cline{4-8}
& & & Orig. & {\scriptsize$\Theta(N^2)$} & 0 & $0$ & $\infty$ \\ \hline
\end{tabular}
}
\end{table}

\bibliographystyle{IEEEtran}
\bibliography{access}  % 

\begin{IEEEbiography}[{\includegraphics[width=1in,height=1.25in,clip,keepaspectratio]{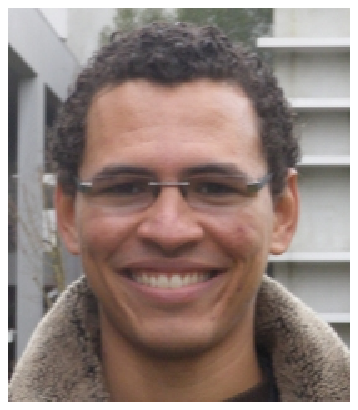}}]{Saulo Queiroz} 
is a permanent lecturer at the Department of Computer Science of 
the Federal University of Technology (UTFPR) in Brazil and Ph.D candidate at the University of Coimbra (Portugal). He completed the B.Sc. (2006) and M.Sc. (2009)  degrees at the  Federal University of Amazonas (Brazil)  with focus on the efficiency of networking algorithms.  Over the last decade, he has lectured disciplines on computer science such as design and analysis of algorithms and signal communication processing. With his research on networking, he has contributed with open source initiatives such as Google Summer of Code. His current research interest comprises the design and analysis of signal communication algorithms.
 \end{IEEEbiography}
\begin{IEEEbiography}[{\includegraphics[width=1in,height=1.25in,clip,keepaspectratio]{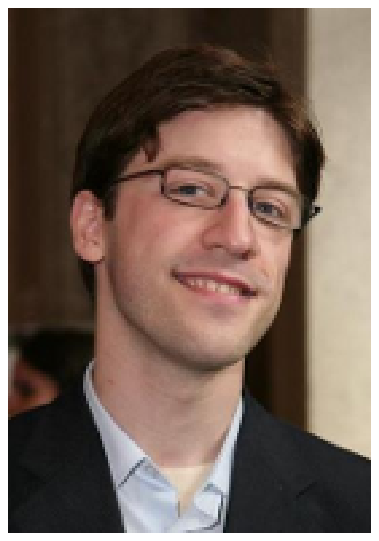}}]{Jo\~ao P. Vilela} 
is a professor at the Department of Computer Science of the University of Porto. He was previously a professor at the Department of Informatics Engineering of the University of Coimbra, after receiving the Ph.D. in Computer Science in 2011 from the University of Porto, Portugal. He was a visiting researcher at the Coding, Communications and Information Theory group at Georgia Tech, working on physical-layer security, and the Network Coding and Reliable Communications group at MIT, working on security for network coding. In recent years, Dr. Vilela has been coordinator and team member of several national, bilateral, and European-funded projects in security and privacy of computer and communication systems, with focus on wireless networks, Internet of Things and mobile devices.
\end{IEEEbiography}
\begin{IEEEbiography}[{\includegraphics[width=1in,height=1.25in,clip,keepaspectratio]{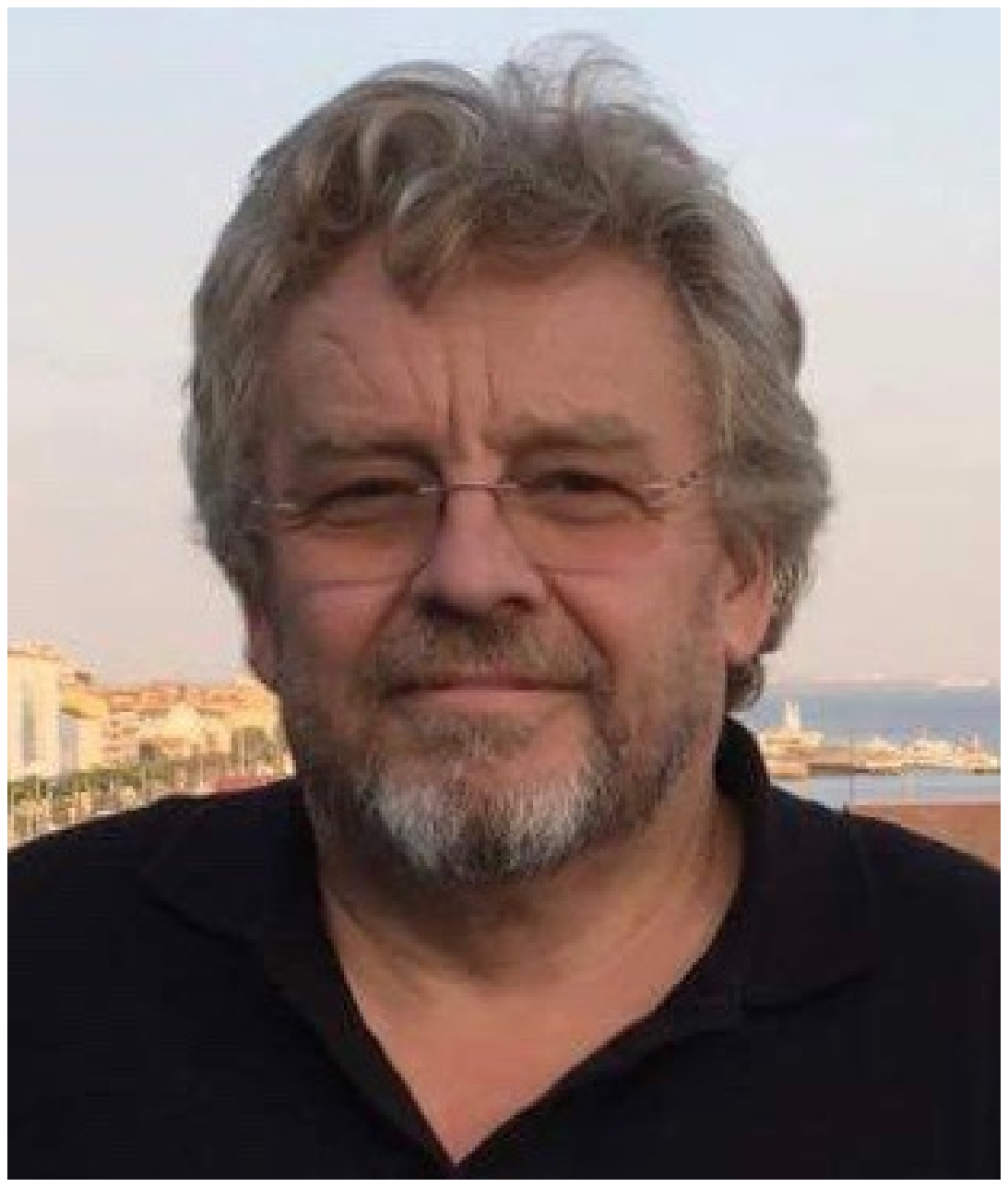}}]{Edmundo Monteiro}
is currently a Full Professor with the University of Coimbra, Portugal.
He has more than 30 years of research experience in the field of computer communications,
wireless networks, quality of service and experience, network and service management, and
computer and network security. He participated in many Portuguese, European, and international
research projects and initiatives. His publication
list includes over 200 publications in journals,
books, and international refereed conferences. He has co-authored nine
international patents. He is a member of the Editorial Board of Wireless
Networks (Springer) journal and is involved in the organization of many
national and international conferences and workshops. He is also a Senior
Member of the IEEE Communications Society and the ACM Special Interest
Group on Communications. He is also a Portuguese Representative in IFIP
TC6 (Communication Systems).
\end{IEEEbiography}

\EOD
% that's all folks
\end{document}